\documentclass[a4paper,USenglish,hyperref,cleveref,autoref,thm-restate,numberwithinsect,nameinlink,pagebackref]{lipics-v2021}
\hideLIPIcs %
\nolinenumbers

\usepackage[T1]{fontenc}
\usepackage{color}

\usepackage[utf8]{inputenc}
\usepackage{amsmath}
\usepackage{amsfonts}
\usepackage{amssymb}
\usepackage{mathtools}
\usepackage{xspace}
\usepackage[numbers,sort]{natbib}
\usepackage[dvipsnames]{xcolor}
\usepackage{todonotes}
\usepackage{dsfont}
\usepackage[shortlabels]{enumitem}
\usepackage{algpseudocode}
\usepackage{algorithm}

\usepackage{cleveref}
\usepackage{calc}
\usepackage{thm-restate}
\usepackage{orcidlink}

\usepackage{tikz}
\usetikzlibrary{matrix,fit,decorations.pathreplacing,decorations.pathmorphing,shapes,patterns,patterns.meta,intersections,positioning,chains,calc,backgrounds,arrows.meta,math,decorations.markings}

\tikzstyle{vertex_small}=[circle,draw,fill, inner sep=1.5pt]
\tikzset{
	position/.style args={#1:#2 from #3}{
		at=(#3.#1), anchor=#1+180, shift=(#1:#2)
	}
}

\definecolor{italyGreen}{RGB}{0, 146, 70}
\definecolor{italyRed}{RGB}{206, 43, 55}

\newlist{compactitem}{itemize}{4}
\setlist[compactitem,1]{nolistsep,label=--}

\DeclareMathOperator{\W}{W}
\DeclareMathOperator{\dist}{dist}
\DeclareMathOperator{\degree}{deg}
\DeclareMathOperator*{\argmin}{arg\,min}

\DeclareMathOperator{\cc}{cc}
\DeclareMathOperator{\internal}{int}
\DeclareMathOperator{\module}{mdl}

\DeclareMathOperator{\abov}{ab}
\DeclareMathOperator{\below}{be}
\DeclareMathOperator{\ml}{ml}

\algnewcommand\algorithmicforeach{\textbf{for each}}
\algdef{S}[FOR]{ForEach}[1]{\algorithmicforeach\ #1 \algorithmicdo}

\newcommand{\problemDef}[3]{%
\begin{center}
	\setlength{\tabcolsep}{2pt}
	\begin{tabular}{@{}lp{10cm}@{}}
		\multicolumn{2}{@{}l}{\textsc{#1}} \\%
		\textbf{Input:} & #2 \\%
		\textbf{Question:} & #3 \\%
	\end{tabular}
\end{center}%
}

\newcommand{\N}{\mathds{N}}
\newcommand{\Z}{\mathds{Z}}

\newcommand{\MADSTl}{\textsc{MAD Spanning Tree}\xspace}
\newcommand{\MADST}{MADST\xspace}
\newcommand{\ExThreeCov}{\textsc{Exact Cover by $3$-Sets}\xspace}
\newcommand{\XtC}{\textsc{X$3$C}\xspace}
\newcommand{\MAD}{MAD tree\xspace}
\newcommand{\MADs}{MAD trees\xspace}

\usepackage{etoolbox}

\newcommand{\appsymb}{$\star$}
\newcommand{\appref}[1]{{\hyperref[proof:#1]{\appsymb}}}
\makeatletter
\let\@mkboth\@gobbletwo
\makeatother


\relatedversion{A conference version of this work was published at IWOCA 2026~\cite{conferenceVersion}.}

\author{Tom-Lukas Breitkopf}{Technische Universität Berlin, Algorithmics and Computational Complexity Research Group, Germany}{t.breitkopf@tu-berlin.de}{https://orcid.org/0009-0008-2875-1945}{}

\author{Vincent Froese}{Technische Universität Berlin, Algorithmics and Computational Complexity Research Group, Germany}{vincent.froese@tu-berlin.de}{https://orcid.org/0000-0002-8499-0130}{}

\author{Anton Herrmann}{Technische Universität Berlin, Algorithmics and Computational Complexity Research Group, Germany}{a.herrmann@tu-berlin.de}{https://orcid.org/0009-0008-8473-9043}{}

\author{André Nichterlein}{Technische Universität Berlin, Algorithmics and Computational Complexity Research Group, Germany}{andre.nichterlein@tu-berlin.de}{https://orcid.org/0000-0001-7451-9401}{}

\author{Camille Richer}{Université Paris-Dauphine, PSL Research University, CNRS, UMR 7243, LAMSADE, Paris, France \&
Orange Research, Châtillon, France}{camille.richer@dauphine.eu}{https://orcid.org/0009-0000-3636-6571}{}

\authorrunning{T.-L.\ Breitkopf, V.\ Froese, A.\ Herrmann, A.\ Nichterlein, C.\ Richer}

\begin{document}

\title{Parameterized Algorithms for\\ Computing MAD Trees}

\titlerunning{Parameterized Algorithms for Computing MAD Trees}

\maketitle

\begin{abstract}
We consider the well-studied problem of finding a spanning tree with minimum average distance between vertex pairs (called a MAD tree).
This is a classic network design problem which is known to be NP-hard.
While approximation algorithms and polynomial-time algorithms for some graph classes are known, the parameterized complexity of the problem has not been investigated so far.

We start a parameterized complexity analysis with the goal of determining the border of algorithmic tractability for the MAD tree problem. 
To this end, we provide a linear-time algorithm for graphs of constant modular width and a polynomial-time algorithm for graphs of bounded treewidth; the degree of the polynomial depends on the treewidth.
That is, the problem is in FPT with respect to modular width and in XP with respect to treewidth.
Moreover, we show that it is in FPT when parameterized by vertex integrity, by maximum leaf number, or by an above-guarantee parameter.
We complement these algorithms with NP-hardness on split graphs.

\end{abstract}

\section{Introduction}
\label{sec:intro}

Computing a spanning tree is a classic algorithmic graph problem in computer science. There are numerous objectives to optimize, for example the weight, the diameter, the radius or the number of leaves of the spanning tree. In this work, we consider the problem of finding a spanning tree of an undirected unweighted graph with \emph{minimum average distance} (called a \emph{MAD tree}), that is, a spanning tree that minimizes the average distance of all vertex pairs. Note that this is equivalent to minimizing the total sum of pairwise vertex distances also known as the Wiener index; an important concept in chemical graph theory \cite{Wiener47}.

The problem of computing a MAD tree (also known as a \emph{minimum routing cost spanning tree}) was introduced by \citet{Hu74} (there called \emph{optimum distance spanning tree}) and has various applications, for example in network design for transportation and communication networks, facility location, and also in genome alignment in biology \cite{FischettiLS02}.
It is known to be NP-hard~\cite{JLK78}.
On the algorithmic side, a polynomial-time approximation scheme was shown by \citet{WLBCRT00}.
Regarding exact algorithms, a cubic-time algorithm is known for series-parallel graphs~\cite{EC85}. The problem is also known to be linear-time solvable on complete $r$-partite graphs~\cite{BE95}, distance-hereditary graphs~\cite{DDGS03}, and interval graphs~\cite{DDR04} and solvable in quadratic time on permutation graphs~\cite{JM12}, trapezoid graphs~\cite{Mondal13} and circular-arc graphs~\cite{JMP15}. 
Moreover, several mixed integer programming formulations were developed \cite{FischettiLS02} (for an overview we refer to \citet{ZCFL19}). Beyond that, many heuristic approaches have been considered (see the survey by \citet{MNSS19}).

We extend these algorithmic results by adopting a \emph{parameterized complexity} view, that is, we develop exact algorithms where the exponential part of the running time can be confined to certain parameters of the input graph. 
This leads to efficient algorithms for inputs where the parameter is small (constant).
Our starting point is an open question by \citet{DDGS03}, who asked whether the problem is polynomial-time solvable on graphs of bounded-cliquewidth, i.e., whether it is in XP with respect to cliquewidth.
While we could not resolve this question, we make progress with larger parameters; see \Cref{fig:mad-trees-hasse-diagram} for an overview.
Concretely, we obtain the following results: The problem is \emph{fixed-parameter tractable} (in FPT) for the parameters vertex integrity, modular width, and maximum leaf number of the input graph, and it is in XP for the treewidth. 
Moreover, we give a stronger hardness result showing NP-hardness on split graphs.

\begin{figure}[t]
    \centering
    \begin{tikzpicture}[scale=1.4]
        \sffamily\small
        \node[align=center, draw=black, rounded corners, fill=italyGreen!50] (dtc) at (0,0) {distance to clique};
        \node[align=center, draw=black, rounded corners, fill=italyGreen!50] (vc) at (2.5,0) {vertex cover};
        \node[align=center, draw=black, rounded corners, fill=italyGreen!50] (mln) at (5,0) {max leaf number\\(\Cref{thrm:ml_fpt})};
        \node[align=center, draw=black, rounded corners] (cvd) at (-1,1) {cluster vertex deletion};
        \node[align=center, draw=black, rounded corners, fill=italyGreen!50] (vi) at (2.5,1) {vertex integrity\\(\Cref{thm:fpt-vertex-integrity})};
        \node[align=center, draw=black, rounded corners, fill=CornflowerBlue!60] (tw) at (2.5,2) {treewidth\\(\Cref{thm:treewidth-XP})};
        \node[align=center, draw=black, rounded corners, fill=italyGreen!50] (mw) at (0.5,1.75) {modular width\\(\Cref{thrm:mw_fpt})};
        \node[align=center, draw=black, rounded corners] (cw) at (2.5,3) {clique-width};
        \node[align=center, draw=black, rounded corners, fill=CornflowerBlue!60] (fes) at (5,1) {feedback edge set};
        \draw[->] (dtc)--(cvd);
        \draw[->] (vc)--(vi);
        \draw[->] (mln)--(fes);
        \draw[->] (vi)--(tw);
        \draw[->] (tw)--(cw);
        \draw[->] (cvd)to[bend left = 35](cw);
        \draw[->] (mw)--(cw);
        \draw[->] (vc)to[bend left=15](mw);
        \draw[->] (vc)--(cvd);
        \draw[->] (dtc.30)--(mw);
        \draw[->] (fes)--(tw);
    \end{tikzpicture}
    \caption{Hasse diagram of the parameterized complexity of \MADST with parameters growing from top to bottom. An arrow from parameter $a$ to $b$ indicates that parameter $b$ is bounded by a function of $a$ on any graph. Green background stands for FPT cases, blue for XP cases.}
    \label{fig:mad-trees-hasse-diagram}
\end{figure}

\subparagraph*{Further Related Work.}
The structure of MAD trees is well studied in the graph theory literature~\cite{BES97,Entringer,Da00,LYTN13}.
They appear in many different contexts, for example in computational geometry~\cite{Abu-AffashCLM24}, distributed computing~\cite{HochuliHW14}, computational biology~\cite{FischettiLS02}, and operations research~\cite{ZCFL19}.
There are also many other variants of special spanning trees of interest, for example, shortest-path trees~\cite{HZ96}, $k$-leaf spanning trees \cite{Ze18}, or minimum-diameter spanning trees \cite{HT95}.

\subparagraph*{Organization.}
After formally introducing the problem and further preliminaries in \cref{sec:preliminaries} we turn towards width parameters in \cref{sec:width-params} giving our results on modular width and treewidth.
To further investigate the border of tractability we turn to parameters larger than treewidth in \cref{sec:larger-params} where we consider the parameter vertex integrity, a problem-specific above guarantee parametrization and the maximum leaf number.
Regarding lower bounds we strengthen known NP-hardness results in \cref{sec:hardness} before concluding in \cref{sec:conclusion}.

\section{Preliminaries}
\label{sec:preliminaries}
For~$n \in \N$, we define~$[n] \coloneqq \{1,2,\ldots,n\}$ with~$[0] \coloneqq \emptyset$.

\subparagraph*{Graphs.}\label{sec:Notation}
All graphs in this work are simple and unweighted.
For a graph~$G$, we denote the set of vertices by~$V(G)$ and the set of edges by~$E(G)$.
For an edge~$\{u,v\} \in E(G)$, we write~$uv$ as a shorthand.
We set~$n_G \coloneqq |V(G)|$ and $m_G \coloneqq |E(G)|$.
We denote the degree of a vertex~$v\in V(G)$ by~$\deg_G(v)\coloneqq |N_G(v)|$, where~$N_G(v)\coloneqq\{u\in V(G)\mid uv \in E(G)\}$.
The set~$N_G(v)$ is called the open neighborhood of~$v$ and~$N_G [v] \coloneqq N_G(v) \cup \{v\}$ is called the closed neighborhood.
If the considered graph is clear from the context, then we drop the subscripts.
We write~$H \subseteq G$ to denote that~$H$ is a subgraph of~$G$.
For a subset~$X\subseteq V(G)$, we write~$G-X$ for the subgraph of~$G$ obtained by deleting the vertices from~$X$ and all edges which have nonempty intersection with~$X$.
For an edge~$uv \in E(G)$, we write~$G-uv$ for the subgraph of~$G$ obtained by deleting~$uv$ from~$G$.

A tree~$T$ is an acyclic and connected graph.
A subgraph~$H \subseteq G$ is a spanning tree of~$G$ if~$H$ is a tree and~$V(H) = V(G)$.
Throughout the rest of this work, we assume~$G$ to be connected since we are interested in spanning trees.
We define the distance~$\dist_G(u,v)$ between the two vertices~$u$ and~$v$ in~$G$ as the length of the shortest~$u$-$v$-path in~$G$.
For a vertex~$v$, we write~$\dist_G(v)\coloneqq \sum_{u \in V(G)} \dist_G(u,v)$.
A vertex~$v$ of~$G$ is said to be \emph{median} if it minimizes~$\dist_G(v)$.
For a vertex set~$A \subseteq V$, we define~$\dist_G(A,A) \coloneqq \frac{1}{2} \sum_{u,v \in A} \dist_G(u,v)$ and if~$B \subseteq V$ is another vertex set with~$A \cap B = \emptyset$, then we define~$\dist_G(A,B) \coloneqq \sum_{u \in A, v \in B} \dist_G(u,v)$.
For an edge~$uv$ of a tree~$T$, we define~$T_{uv}^u$ to be the subtree of~$T$ that contains~$u$ when removing the edge~$uv$ from~$T$.
Finally, $P_{u,v}^T$ denotes the unique path between vertices~$u$ and~$v$ in the tree~$T$ and if~$T$ is clear from the context we drop the superscript.
Two graphs~$G$ and~$H$ are called isomorphic if there exists a bijection~$\phi \colon V(G) \to V(H)$ such that~$uv \in E(G) \iff \phi(u)\phi(v) \in E(H)$.
In this case, $\phi$ is called an isomorphism and we write~$G \cong H$.

\subparagraph*{Wiener Index.}
The Wiener index of a graph~$G$ is the sum of all pairwise distances, that is, $\W(G)\coloneqq\sum_{\{u,v\}\subseteq V}\dist_G(u,v)$.
Several formulas for the Wiener index of a tree~$T$ are known~\cite{Schmuck10}.
In particular, we heavily use the following formula:
For an edge~$uv$, let~$w_T(uv) \coloneqq |V(T_{uv}^u)| \cdot |V(T_{uv}^v)|$.
Then~$\W(T) = \sum_{uv \in E(T)} w_T(uv)$.
 A spanning tree with minimum Wiener index is called a \emph{minimum average distance} (MAD) tree.
\citet{DDR04} proved the following useful lemma about median vertices in MAD trees.
    \begin{lemma}[\citet{DDR04}]
        \label{lem:Dankelmann}
        If~$c$ is a median vertex of a MAD tree~$T$ of a graph~$G$, then every path in~$T$ starting at~$c$ is an induced path in~$G$. Moreover,
         the tree $T$ and the median~$c$ can be chosen such that there is no vertex~$c' \neq c$ with~$N_G[c]\subset N_G[c']$.
    \end{lemma}
 
We study the following decision problem:
\problemDef{\MADSTl (\MADST)}
{An undirected connected graph $G=(V,E)$ and an integer~$b$.}
{Is there a spanning tree~$T$ of~$G$ with~$\W(T)\le b$?}

\subparagraph*{Parameterized Complexity.}
We assume the reader to be familiar with basic notions from complexity theory like NP-hardness.
Parameterized complexity is a multivariate approach to measure the time complexity of computational problems~\cite{DF13,CFKLMPPS15}. 
An instance~$(x,k)$ of a parameterized problem consists of a classical instance~$x$ together with a number~$k$ called \emph{parameter}. 
A parameterized problem is called \emph{fixed-parameter tractable} (that is, it is contained in the class FPT) if there is an algorithm deciding an instance~$(x,k)$ in time~$f(k)\cdot|x|^{O(1)}$, where~$f$ is a computable function solely depending on~$k$. 
The class~XP contains all parameterized problems which are polynomial-time solvable for constant parameter values, that is, in time $|x|^{f(k)}$. It is known that FPT $\subsetneq$ XP.

\section{Width Parameters}
\label{sec:width-params}

In this section, we analyze the complexity of \MADST with respect to treewidth and modular width.
Our first result is a dynamic program running in polynomial time on graphs of bounded treewidth. Note that this generalizes for example the polynomial-time algorithm on series-parallel graphs \cite{EC85}, which have treewidth two. 
As a second result we give a branching algorithm for \MADST parameterized by modular width.

\subsection{Treewidth}

Informally, the treewidth of a graph reflects how close the graph is to being a tree.
A \emph{tree-decomposition} of a graph $G$ is a pair $(D, \{B_t : t \in V(D) \})$ where $D$ is a rooted tree and each node $t$ of $D$ has an associated \emph{bag} $B_t \subseteq V(G)$. (To avoid confusion with the vertices of $G$, we refer to a vertex of $D$ as a \emph{node}.) The bags must satisfy the following conditions: 
\begin{enumerate}
    \item $\bigcup_{t \in V(D)} B_t = V(G)$;
    \item for each $uv \in E(G)$, there is a node $t \in V(D)$ such that $\{u, v\} \subseteq B_t$;
    \item for each $v \in V(G)$, the set of nodes $\{t \in V(D) : v \in B_t \}$ induces a connected subtree of~$D$.
\end{enumerate}

The \emph{width of a tree-decomposition} is $\max_{t \in V(D)} |B_t| - 1$. The \emph{treewidth} of a graph $G$ is the minimum width over all possible tree decompositions of $G$.
The ``$-1$'' in the definition of the width of a tree-decomposition is included for the convenience that trees have treewidth~$1$.

In practice, to design algorithms, we use a special type of tree-decomposition called \emph{nice tree-decomposition} \cite{CFKLMPPS15}. A tree-decomposition $(D, \{B_t : t \in V(D) \})$ of a graph $G$ is \emph{nice} if all of its nodes satisfy one of the following:
\begin{enumerate}
    \item the \emph{root node} $r$ has $|B_r|=1$; 
    \item a \emph{leaf node} $t$ has $|B_t|=1$;
    \item an \emph{introduce node} $t$ has a unique child $t'$ such that $B_t = B_{t'} \cup \{u\}$ for some $u \in V(G)$;
    \item a \emph{forget node} $t$ has a unique child $t'$ such that $B_{t'} = B_t \cup \{u\}$ for some $u \in V(G)$;
    \item a \emph{join node} $t$ has exactly two children $t'$ and $t''$ such that $B_t=B_{t'}=B_{t''}$. 
\end{enumerate}
Given a node $t$ in $D$, we let $Y_t$ be the set of all vertices contained in the bags of the subtree of $D$ rooted at $t$.

Let $(G,b)$ be an instance of \MADST.
Assume we have a nice tree decomposition $(D, \{B_t \colon t \in V(D)\})$ of $G$.
We use a bottom-up dynamic programming approach on $(D, \{B_t \colon t \in V(D)\})$ to compute a minimum average distance tree of $G$. 
Whenever an edge~$uv$ of the graph is seen for the last time, we compute its contribution~$w_T(uv)$ to each potential solution~$T$.
Recall that~$w_T(uv) \coloneqq |V(T_{uv}^u)| \cdot |V(T_{uv}^v)|$ where~$|V(T_{uv}^u)|$ denotes the number of vertices in the tree~$T$ that are closer to~$u$ than to~$v$.
As for a tree~$T$ the Wiener index is~$\W(T) = \sum_{uv \in E(T)} w_T(uv)$, we can simply add all the contributions to a potential solution~$T$ and store the cost of all partial solutions in a table.
Here, consistency in keeping the number of vertices closer to one vertex than a neighboring vertex is the technically involved part responsible for the majority of information we store in the table.
We achieve the following:
\begin{restatable}{theorem}{treewidthXP}
	\label{thm:treewidth-XP}
	\MADST is solvable in $2^{O(2^k)}n^{O(k)}$ time on treewidth-$k$ graphs.
\end{restatable}

For the following definitions, consider a spanning tree $T$ of $G$ and a node $t$ of $D$.
Observe that any bag $B_t$ of the nice tree decomposition is a separator of $G$.
It partitions the vertices of $G$ into three sets: the vertices in $Y_t \setminus B_t$ have been introduced and forgotten in the subtree of $D$ rooted at $t$, the vertices of $B_t$ have been introduced but not forgotten, and the vertices of $V \setminus Y_t$ have not yet been introduced.
In a similar fashion, it partitions the edges of a spanning tree of $G$. This is formalized by the following three definitions.

\begin{definition}
    An edge of $T$ is \emph{forgotten} in $t$ if at least one of its endpoints is in $Y_t \setminus B_t$. It is \emph{hidden} in $t$ if at least one of its endpoints is in $V(G) \setminus Y_t$.
\end{definition}

An edge of $T$ is either forgotten, hidden, or in the forest $T[B_t]$.
The contribution of an edge $uv$ to the cost of $T$ is computed when $uv$ is forgotten, that is, when its first endpoint (say $v$) is forgotten.
To compute the contribution of the edge $uv$, we only need the number of vertices of $T$ closer to the endpoint that is not yet forgotten (here $u$).
These are the vertices that are connected to $u$ in $T$ by a path that does not contain $uv$. 
That path begins either with a forgotten edge, or with a hidden edge, or with an edge of the forest $T[B_t]$.

\begin{definition}
    Let $u \in B_t$. Another vertex $v \in V(G)$ is \emph{below} $u$ in $T$ if the path from $u$ to $v$ in $T$ starts with a forgotten edge. And $v$ is \emph{above} $u$ in $T$ if the path from $u$ to $v$ in $T$ starts with a hidden edge.
\end{definition}

For any connected component in the forest $T[B_t]$, a vertex of $V(G)$ is either in the connected component, or below or above some vertex of the connected component.

Finally, the definition of a \emph{below connection} is useful to avoid creating cycles when enumerating the forests in $B_t$.

\begin{definition}
    Consider a maximal subtree of $T[Y_t]$ containing only forgotten edges such that no vertex of $B_t$ is an internal vertex. Let $C$ be the leaves of the subtree that are in $B_t$. If $|C|\geq 2$, we say that the vertices in $C$ are \emph{connected from below} in $T$ or that $C$ is a \emph{below connection} in $T$.
\end{definition}

Between any two vertices of a below connection, there is a path consisting only of forgotten edges and vertices of $Y_t \setminus B_t$.
See \Cref{fig:tw-def} for an example.
\begin{figure}[t]
    \centering
    \begin{tikzpicture}[xscale=1.5, yscale=1.2]
        \node[vertex_small] (a) at (1,0.25) [label=left:{\footnotesize$a$}] {};
        \node[vertex_small] (b) at (1.5,0.5) [label=above left:{\footnotesize$b$}] {};
        \node[vertex_small] (c) at (2.5,0) [label=above right:{\footnotesize$c$}] {};
        \node[vertex_small] (d) at (2.5,1) [label=left:{\footnotesize$d$}] {};
        \node[vertex_small] (e) at (3,0.5) [label=above:{\footnotesize$e$}] {};
        \node[vertex_small] (f) at (3.5,0) [label=left:{\footnotesize$f$}] {};
        \node[vertex_small] (g) at (4,1) [label=below:{\footnotesize$g$}] {};
        
        \draw (c)--(d) (e)--(f)--(g);
        \node[vertex_small, color=CornflowerBlue] (f1) at (0.75,-1.25) {};
        \node[vertex_small, color=CornflowerBlue] (f2) at (1,-0.75) {};
        \node[vertex_small, color=CornflowerBlue] (f3) at (2,-1) {};
        \node[vertex_small, color=CornflowerBlue] (f4) at (2.5,-1.5) {};
        \node[vertex_small, color=CornflowerBlue] (f5) at (3,-1) {};
        \node[vertex_small, color=CornflowerBlue] (f6) at (3.5,-0.75) {};
        \node[vertex_small, color=CornflowerBlue] (f7) at (4.25,-0.5) {};
        \node[vertex_small, color=CornflowerBlue] (f8) at (4.25,-1) {};
        \draw[CornflowerBlue] (f1)--(f2)--(f3)--(f4) (f7)--(f8) (a)--(f2) (b)--(f3) (c)--(f3) (c)--(f5) (f)--(f6) (f)--(f7);
        \node[vertex_small, color=Apricot] (h1) at (1.5,1.5) {};
        \node[vertex_small, color=Apricot] (h2) at (3,1.75) {};
        \node[vertex_small, color=Apricot] (h3) at (3.25,2) {};
        \node[vertex_small, color=Apricot] (h4) at (3.5,1.75) {};
        \node[vertex_small, color=Apricot] (h5) at (4,2) {};
        \draw[Apricot] (h2)--(h3)--(h4) (b)--(h1) (d)--(h2) (f)--(h2) (g)--(h5);
        \begin{scope}[on background layer]
            \draw[draw=none, fill=black!10, rounded corners=15pt] (0.6, -0.3) rectangle (4.4, 1.3);
        \end{scope}

        \draw [decorate, decoration = {brace}] (4.5,2.25) -- node [midway, right] {$\;$ hidden vertices and edges} (4.5,1.35);
        \draw [decorate, decoration = {brace}] (4.5,1.3) -- node [midway, right] {$\;$ $T[B_t]$} (4.5,-0.3);
        \draw [decorate, decoration = {brace}] (4.5,-0.35) -- node [midway, right] {$\;$ forgotten vertices and edges} (4.5,-1.25);

        \begin{scope}[shift={(1.625,-0.5)}, scale=1.2, shift={(-1.625,0.5)}]
            \draw[rounded corners, dotted] (1,0.25)--(1.5,0.5)--(2.5,0)--(2.5,-1.5)-- (0.75,-1.25)--cycle;
        \end{scope}
        \begin{scope}[shift={(2.75,-0.5)}, scale=1.2, shift={(-2.75,0.5)}]
            \draw[rounded corners, dotted] (2.35,0)--(2.6,0.1)--(3.15,-1)--(2.9,-1.1)--cycle;
        \end{scope}
        \begin{scope}[shift={(3.5,-0.375)}, scale=1.2, shift={(-3.5,0.375)}]
            \draw[rounded corners, dotted] (3.4,0.1)--(3.6,0.1)--(3.6,-0.8)--(3.4,-0.8)--cycle;
        \end{scope}
        \begin{scope}[shift={(3.875,-0.5)}, scale=1.2, shift={(-3.875,0.5)}]
            \draw[rounded corners, dotted] (3.4,-0.1)--(3.6,0.1)--(4.35,-0.5)--(4.3,-1)--(4.1,-1)--cycle;
        \end{scope}
    \end{tikzpicture}
    \caption{The four subtrees of forgotten vertices are highlighted. The set $\{a,b,c\}$ of leaves of the leftmost one is the only below connection in $T$. There are no vertices above or below $e$, there are 13 above and 3 below $f$, and 1 above and none below~$g$.}
    \label{fig:tw-def}
\end{figure}
In the dynamic programming algorithm, we use the formula over edges to compute the cost of a solution: $\W(T)=\sum\limits_{uv \in E(T)} w_{T}(uv)$.

\begin{definition}
    The \emph{partial cost} of $T$ in $t$ is the cost of its forgotten edges according to the formula above: $$\sum\limits_{uv \in E(T[Y_t]) \setminus E(T[B_t])} w_{T}(uv).$$
\end{definition}

\begin{proof}[Proof of \cref{thm:treewidth-XP}]
We provide a dynamic programming algorithm on the nice tree decomposition of $G$. For each node $t$ of $D$, we keep a table $\mathcal{T}_t$ indexed by: \begin{itemize}
    \item a set of edges $F \subseteq E(G[B_t])$,
    \item a collection $\mathcal{C}$ of subsets of $B_t$ where each subset has size at least 2,
    \item a function $\abov \colon B_t \to \{0, \ldots, n_G-1\}$,
    \item a function $\below \colon B_t \to \{0, \ldots, n_G-1\}$.
\end{itemize}

The value stored in $\mathcal{T}_t[F,\mathcal{C},\abov,\below]$ is the minimum partial cost of a spanning tree $T$ of $G$ satisfying the following: 
\begin{itemize}
    \item The edges of the forest $T[B_t]$ are exactly $F$.
    \item $\mathcal{C}$ is exactly the collection of below connections with respect to $t$ in $T$.
    \item For every $u \in B_t$, $\abov(u)$ is equal to the number of vertices above $u$ in $T$.
    \item For every $u \in B_t$, $\below(u)$ is equal to the number of vertices below $u$ in $T$.
\end{itemize}

If there is no solution satisfying the conditions imposed by $(F,\mathcal{C}, \abov, \below)$, then we set $\mathcal{T}_t[F,\mathcal{C},\abov,\below]\coloneqq\infty$.
Note that $\mathcal{T}_t[F,\mathcal{C},\abov,\below] \neq \infty$ implies: \begin{enumerate}
    \item The set $F$ induces a forest on $B_t$.
    \item For every $C \in \mathcal{C}$, $C$ contains at most one vertex from each connected component of the forest induced by $F$.
    \item For each connected component $H$ of the forest induced by $F$, it holds \begin{align*}
        \sum\limits_{u \in V(H)} (1+\abov(u)+\below(u)) = n_G.
    \end{align*}
\end{enumerate}
We denote these necessary conditions by~$(\star)$.
If $(F,\mathcal{C},\abov,\below)$ does not satisfy $(\star)$, then there is no solution for $(F,\mathcal{C},\abov,\below)$ and the corresponding table entry is automatically set to $\infty$.

The cost of an optimal solution will be retrieved at the end of the algorithm from the table of the root node $r$ as $\mathcal{T}_r[\emptyset, \emptyset, \abov,\below]$, where $\abov(v)=0$ and $\below(v)=n_G-1$ (with $B_r=\{v\}$).

We now describe how to process each node of the tree decomposition.

\paragraph*{Leaf node.}
Let $t$ be a leaf node with $B_t=\{u\}$. 
In any spanning tree of $G$, all edges incident to $u$ are hidden in $t$ and all other vertices of $V(G)$ are above $u$.
Thus, we set
\begin{equation*}
    \mathcal{T}_t[\emptyset, \emptyset, \abov,\below]\coloneqq
    \begin{cases}
      0, & \text{if}\ \abov(u)=n_G-1 \ \text{and}\ \below(u)=0 \\
      \infty, & \text{otherwise.}
    \end{cases}
\end{equation*}

\paragraph*{Introduce node.}
Let $t$ be an introduce node introducing vertex $u$ and let $t'$ denote its child. 
Let $(F,\mathcal{C}, \abov, \below)$ satisfy $(\star)$ with respect to $t$.
Since $t$ introduces $u$, there is no edge between $u$ and $Y_{t} \setminus B_t$ in $G$, so in any spanning tree $T$ of $G$, there should be no vertex below $u$ in $T$ with respect to~$t$.
So, if $u \in C$ for some $C \in \mathcal{C}$ or if $\below(u) \neq 0$, then we can set $\mathcal{T}_t[F,\mathcal{C}, \abov, \below] \coloneqq \infty$.
We now assume that $(F,\mathcal{C}, \abov, \below)$ satisfies $(\star)$, $u \notin C$ for all $C \in \mathcal{C}$ and $\below(u)=0$.

Then 
\begin{align}\label{eq:tw_introduce}
    \mathcal{T}_t[F,\mathcal{C},\abov,\below] = \mathcal{T}_{t'}[F \setminus \{e : u \in e\},\mathcal{C},\abov',\below'],
\end{align}
where $\below'$ is the restriction of $\below$ to $B_{t'}$ and $\abov'\colon B_{t'}\to\{0,\ldots,n_G-1\}$ is given by
\begin{equation*}
    \abov'(v)\coloneqq
    \begin{cases}
      \abov(v), & \text{if}\ uv \notin F \\
      \abov(v)+\sum\limits_{x \in V(F^{u}_{uv})} (1+\abov(x)+\below(x)), & \text{otherwise}
    \end{cases},
\end{equation*}
where $F^{u}_{uv}$ denotes the connected component containing $u$ in the forest induced by $F \setminus \{uv\}$.

\begin{proof}[Proof of \Cref{eq:tw_introduce}]
    Consider a spanning tree $T$ of $G$ satisfying the conditions $(F, \mathcal{C}, \abov, \below)$ with respect to $t$, where $\below(u)=0$ and $u \notin C$ for all $C \in \mathcal{C}$. See \Cref{fig:tw-intro-node} for an illustration. 
    
    Since no edge is forgotten between $t'$ and $t$, the partial cost of $T$ is the same with respect to $t$ and $t'$. 
    The edges of the forest $T[B_{t'}]$ are those of $F$ except those containing $u$ which are hidden with respect to $t'$.
    
    No edge or vertex is forgotten, so below connections in $T$ are identical with respect to~$t$ and $t'$ and for every $v \in B_{t'} \subset B_t$, the same vertices are below $v$ in $T$ with respect to $t'$ and~$t$.
    
    Let $v \in B_{t'}$.
    If $v$ is not a neighbor of $u$ in $T$, then the hidden edges incident to $v$ are identical with respect to $t$ and $t'$, so the same vertices are above $v$ in $T$ with respect to~$t$ and $t'$.
    If $v$ is a neighbor of $u$ in $T$ then the edge $uv$ is hidden with respect to $t'$ but is in~$E(T[B_t])$.
    All the vertices in the connected component of $T-uv$ containing $u$ are above~$v$ in~$T$ with respect to $t'$ but not with respect to $t$. 
    Let $T[B_t]^{u}_{uv}$ be the connected component of~$T[B_t] - uv$ containing $u$.
    All these vertices above $v$ with respect to $t'$ but not to $t$ are either in $T[B_t]^{u}_{uv}$, or above or below some vertex of~$T[B_t]^{u}_{uv}$, so their number is given by 
    \begin{align*}
        \sum\limits_{x \in V(T[B_t]^{u}_{uv})} 1+\abov(x)+\below(x).
    \end{align*}
    The number of vertices above $v$ in $T$ with respect to $t'$ is then
    \begin{align*}
        \abov(v)+\sum\limits_{x \in V(T[B_t]^{u}_{uv})} 1+\abov(x)+\below(x).
    \end{align*}
\end{proof}

\begin{figure}[t]
    \centering
    \begin{tikzpicture}[xscale=1.2]
        \begin{scope}[on background layer]
            \draw[draw=none, fill=black!10, rounded corners=15pt] (0.75, -0.5) rectangle (5.25, 1.75);
        \end{scope}
        
        \draw[loosely dashed] (0.75,1.2)--(5.25,1.2);
        \node[vertex_small] (u) at (3,1.5) [label=left:{$u$}] {};
        \node[vertex_small] (v1) at (2,0) [label=left:{$v$}] {};
        \node[vertex_small] (v2) at (2.5,0) {};
        \node[vertex_small] (v3) at (3.5,0) {};
        \node[vertex_small] (a) at (1,0) {};
        \node[vertex_small] (b) at (1.5,0.75) {};
        \node[vertex_small] (c) at (3,0) {};
        \node[vertex_small] (d) at (3,0.75) {};
        \node[vertex_small] (e) at (4,0) {};
        \node[vertex_small] (f) at (4,0.75) {};
        \node[vertex_small] (g) at (4.5,0.5) {};
        \node[vertex_small] (h) at (5,0.5) {};
        \draw (a)--(b) (v1)--(u)--(v2) (u)--(v3) (v2)--(c)--(d) (e)--(v3)--(f);
        \node[vertex_small, color=CornflowerBlue] (f1) at (1,-1) {};
        \node[vertex_small, color=CornflowerBlue] (f2) at (2,-1) {};
        \node[vertex_small, color=CornflowerBlue] (f3) at (3,-0.75) {};
        \node[vertex_small, color=CornflowerBlue] (f6) at (5,-1) {};
        \draw[CornflowerBlue] (a)--(f1) (v1)--(f2) (f2)to[bend right=5](f6) (f6)--(g) (h)--(f6) (c)--(f3);
        \node[vertex_small, color=Apricot] (h1) at (1.5,2) {};
        \node[vertex_small, color=Apricot] (h2) at (2.5,2) {};
        \node[vertex_small, color=Apricot] (h3) at (4,2) {};
        \node[vertex_small, color=Apricot] (h4) at (4.5,2) {};
        \draw[Apricot] (h1)--(h2) (a)--(h1) (u)--(h2) (f)--(h3) (g)--(h4);
        \draw [decorate, decoration = {brace}] (5.5,1.2) -- node [midway, right] {$B_{t'}$} (5.5,-0.5);
        \draw [decorate, decoration = {brace}] (6.1,1.75) --  (6.1,-0.5) node [midway,right] {$B_t$};
        \begin{scope}[shift={(3.25,0.75)}, scale=1.2, shift={(-3.25,-0.75)}]
            \draw[rounded corners, dotted] (3,1.6)--(4,0.85)--(4,0)--(2.45,0)--cycle;
        \end{scope}
    \end{tikzpicture}

    \caption{The vertex $u$ is introduced in $t$ and thus only has neighbors in $B_{t'}$ or above. The connected component $T[B_t]^{u}_{uv}$ of the forest $T[B_t]-uv$ containing $u$ is enclosed by a dotted contour. The number of vertices above $v$ with respect to $t$ is 0 while it is 14 with respect to $t'$.}
    \label{fig:tw-intro-node}
\end{figure}

\paragraph*{Join node.}
Let $t$ be a join node with children $t'$ and $t''$.

Let $(F,\mathcal{C}, \abov, \below)$ satisfy $(\star)$. Then \begin{align}\label{eq:tw_join}
    \mathcal{T}_t[F,\mathcal{C}, \abov, \below]= \min_{\substack{\mathcal{C}' \subseteq \mathcal{C} \\ 0 \leq \below' \leq \below}} \{\mathcal{T}_{t'}[F,\mathcal{C}', \abov+\below-\below', \below'] + \mathcal{T}_{t''}[F,\mathcal{C} \setminus \mathcal{C}', \abov+\below', \below-\below'] \}
\end{align}

\begin{proof}[Proof of \Cref{eq:tw_join}]
    Consider a spanning tree $T$ of $G$ satisfying the conditions $(F, \mathcal{C}, \abov, \below)$ with respect to $t$.
    See \Cref{fig:tw-join-node} for an illustration.

    An edge forgotten with respect to $t$ has been forgotten in exactly one of the two subtrees rooted in the children $t'$ and $t''$.
    So the minimum partial cost of a solution $\mathcal{T}_t[F, \mathcal{C}, \abov, \below]$ is obtained by summing two minimum partial costs stored in the children: 
    \begin{align*}
        \mathcal{T}_t[F,\mathcal{C}, \abov, \below]= \mathcal{T}_{t'}[F', \mathcal{C}', \abov', \below'] + \mathcal{T}_{t''}[F'', \mathcal{C}'', \abov'', \below'']
    \end{align*}
    for some suitable indices $(F', \mathcal{C}', \abov', \below')$ and $(F'', \mathcal{C}'', \abov'', \below'')$.

    Since $T[B_{t'}]=T[B_{t''}]=T[B_t]$, we should have $F'=F''=F$.
    Consider a below connection in $T$ with respect to $t$. 
    By definition, all the forgotten edges of the associated subtree in $T$ are forgotten in the same subtree rooted either in $t'$ or in $t''$, so we should have $\mathcal{C}''=\mathcal{C} \setminus \mathcal{C}'$.

    Let $u \in B_t$ and $v \in V$. 
    If $v$ is below $u$ in $T$ with respect to $t$, the path from $u$ to $v$ in $T$ starts with an edge forgotten with respect to $t$.
    If that edge was forgotten in the subtree rooted in $t'$, then it is a hidden edge with respect to $t''$, so $v$ is above $u$ with respect to $t''$.
    Conversely, if $v$ is below $u$ with respect to $t''$, then $v$ is above $u$ with respect to $t'$.
    If $v$ is above $u$ in $T$ with respect to $t$, then the path from $u$ to $v$ in $T$ starts with an edge hidden with respect to $t$, so that edge is also hidden with respect to $t'$ and $t''$.
    So for each $u \in B_t$, we should have
    \begin{align*}
        & 0 \leq \below'(u) \leq \below(u) \\
        & \below''(u)=\below(u)-\below'(u) \\
        & \abov'(u)=\abov(u)+\below''(u)=\abov(u)+\below(u)-\below'(u) \\
        & \abov''(u)=\abov(u)+\below'(u).
    \end{align*}

    With the above relations, for every $\mathcal{C}' \subseteq \mathcal{C}$ and for every function $\below' \colon B_{t'} \to \{0, \ldots, n_G-1 \}$ satisfying $0 \leq \below'(u) \leq \below(u)$ for each $u \in B_{t'}$, we get a unique pair of indices $(F', \mathcal{C}', \abov', \below')$ and $(F'', \mathcal{C}'', \abov'', \below'')$.
    We keep $\mathcal{T}_t[F,\mathcal{C}, \abov, \below]$ as the minimum partial cost over every possibility for $\mathcal{C}'$ and $\below'$.
\end{proof}

\begin{figure}[t]
    \centering
    \begin{tikzpicture}
        \scriptsize
        \begin{scope}[on background layer]
            \draw[draw=none,fill=black!10,rounded corners=15pt] (0.375, -0.5) rectangle (4.875, 1);
        \end{scope}
        \node[vertex_small] (a) at (0.75,0) [label=above:{$a$}] {};
        \node[vertex_small] (b) at (2.25,0.5) [label=below:{$b$}] {};
        \node[vertex_small] (c) at (1.5,0) [label=above:{$c$}] {};
        \node[vertex_small] (d) at (3,0) [label=above:{$d$}] {};
        \node[vertex_small] (e) at (3.75,0.5) [label=below:{$e$}] {};
        \node[vertex_small] (f) at (4.5,0) [label=right:{$f$}] {};
        \draw (b)--(c) (d)--(e)--(f);
        \node[vertex_small, color=CornflowerBlue] (j) at (1.125,-1) [label=below:{$j$}] {};
        \node[vertex_small, color=CornflowerBlue] (k) at (3.75,-1) [label=below:{$k$}] {};
        \node[vertex_small, color=CornflowerBlue] (l) at (4.5,-1) [label=below:{$l$}] {};
        \node[vertex_small, color=CornflowerBlue] (m) at (5.25,-1) [label=below:{$m$}] {};
        \draw[CornflowerBlue] (a)--(j)--(c) (f)--(k) (f)--(l) (f)--(m);
        \node[vertex_small, color=Apricot] (g) at (2.25,1.5) [label=above:{$g$}] {};
        \node[vertex_small, color=Apricot] (h) at (3,1.5) [label=above:{$h$}] {};
        \node[vertex_small, color=Apricot] (i) at (3.75,1.5) [label=above:{$i$}] {};
        \draw[Apricot] (b)--(g) (h)--(e) (g)--(h)--(i);

        \begin{scope}[xshift=-4cm,yshift=-4cm]
        \begin{scope}[on background layer]
            \draw[draw=none,fill=black!10,rounded corners=15pt] (0.375, -0.5) rectangle (4.875, 1);
        \end{scope}
        \node[vertex_small] (a') at (0.75,0) [label=above:{$a$}] {};
        \node[vertex_small] (b') at (2.25,0.5) [label=below:{$b$}] {};
        \node[vertex_small] (c') at (1.5,0) [label=above:{$c$}] {};
        \node[vertex_small] (d') at (3,0) [label=above:{$d$}] {};
        \node[vertex_small] (e') at (3.75,0.5) [label=below:{$e$}] {};
        \node[vertex_small] (f') at (4.5,0) [label=right:{$f$}] {};
        \draw (b')--(c') (d')--(e')--(f');
        \node[vertex_small, color=CornflowerBlue] (j') at (1.125,-1) [label=below:{$j$}] {};
        \node[vertex_small, color=CornflowerBlue] (l') at (4.5,-1) [label=below:{$l$}] {};
        \node[vertex_small, color=CornflowerBlue] (m') at (5.25,-1) [label=below:{$m$}] {};
        \draw[CornflowerBlue] (a')--(j')--(c') (f')--(l') (f')--(m');
        \node[vertex_small, color=Apricot] (g') at (2.25,1.5) [label=above:{$g$}] {};
        \node[vertex_small, color=Apricot] (h') at (3,1.5) [label=above:{$h$}] {};
        \node[vertex_small, color=Apricot] (i') at (3.75,1.5) [label=above:{$i$}] {};
        \node[vertex_small, color=Apricot] (k') at (4.5,1.5) [label=above:{$k$}] {};
        \draw[Apricot] (b')--(g') (h')--(e') (g')--(h')--(i') (f')--(k');
        \end{scope}
        \begin{scope}[xshift=4cm,yshift=-4cm]
        \begin{scope}[on background layer]
            \draw[draw=none,fill=black!10,rounded corners=15pt] (0.375, -0.5) rectangle (4.875, 1);
        \end{scope}
        \node[vertex_small] (a'') at (0.75,0) [label=below:{$a$}] {};
        \node[vertex_small] (b'') at (2.25,0.5) [label=below:{$b$}] {};
        \node[vertex_small] (c'') at (1.5,0) [label=below:{$c$}] {};
        \node[vertex_small] (d'') at (3,0) [label=above:{$d$}] {};
        \node[vertex_small] (e'') at (3.75,0.5) [label=below:{$e$}] {};
        \node[vertex_small] (f'') at (4.5,0) [label=right:{$f$}] {};
        \draw (b'')--(c'') (d'')--(e'')--(f'');
        \node[vertex_small, color=CornflowerBlue] (k'') at (3.75,-1) [label=below:{$k$}] {};
        \draw[CornflowerBlue] (f'')--(k'');
        \node[vertex_small, color=Apricot] (g'') at (2.25,1.5) [label=above:{$g$}] {};
        \node[vertex_small, color=Apricot] (h'') at (3,1.5) [label=above:{$h$}] {};
        \node[vertex_small, color=Apricot] (i'') at (3.75,1.5) [label=above:{$i$}] {};
        \node[vertex_small, color=Apricot] (j'') at (1.125,1.5) [label=above:{$j$}] {};
        \node[vertex_small, color=Apricot] (l'') at (4.5,1.5) [label=above:{$l$}] {};
        \node[vertex_small, color=Apricot] (m'') at (5.25,1.5) [label=above:{$m$}] {};
        \draw[Apricot] (a'')--(j'')--(c'') (b'')--(g'') (h'')--(e'') (g'')--(h'')--(i'') (f'')--(l'') (f'')--(m'');
        \end{scope}
    \end{tikzpicture}
    \caption{A solution $T$ with respect to a join node (at the top) and its children (below). The below connection $\{a,c\}$ is also a below connection in the left child but not in the right child. There are three vertices below $f$ in the parent node, two of which are below $f$ in the left node and the last one is below $f$ in the right node.}
    \label{fig:tw-join-node}
\end{figure}

\paragraph*{Forget node.}
Let $t$ be a forget node forgetting vertex $u$ and let $t'$ denote the child of $t$.
Let $(F,\mathcal{C}, \abov, \below)$ satisfy $(\star)$ with respect to $t$. Then \begin{align}\label{eq:tw_forget}
    \mathcal{T}_t[F,\mathcal{C}, \abov, \below]= \min\limits_{\substack{\text{suitable} \\ (F',\mathcal{C}', \abov', \below')}} \{ \mathcal{T}_{t'}[F',\mathcal{C}', \abov', \below'] + \omega \}
\end{align}
where $\omega$ is the cost of the edges forgotten when $u$ is forgotten.

Since $u$ is forgotten by $t$, there is no edge between $u$ and $V \setminus Y_{t'}$, so in any spanning tree $T$ of $G$ there should be no vertex above $u$ in $T$ with respect to $t'$. So in any suitable $(F',\mathcal{C}', \abov', \below')$, we should have $\abov'(u)=0$.

Let $T$ be a spanning tree of $G$ satisfying $(F,\mathcal{C}, \abov, \below)$ with respect to $t$.
Since there is no vertex above $u$ in $T$ with respect to $t'$, $u$ must be connected to vertices of $B_t$ with edges or below connections in $T$ with respect to $t'$.
We split the set of suitable $(F',\mathcal{C}', \abov', \below')$ into three cases according to how $u$ is connected to vertices of $B_t$: \begin{itemize}
    \item there is a unique below connection $\{u,v\}$ between $u$ and some other vertex $v \in B_t$ in $T$ with respect to $t'$;
    \item there is a unique edge $uv$ between $u$ and some vertex $v \in B_t$ in $T$ with respect to $t'$
    \item there are several edges or below connections between $u$ and vertices of $B_t$ in $T$ with respect to $t'$.
\end{itemize}
Each case has its own recursive formula, so we will compute the minimum for each case and then keep the best among these three.

\subparagraph{Case 1:} Among the solutions where $u$ is connected to a unique vertex $v \in B_t$ via a below connection $\{u,v\} \in \mathcal{C}'$ with respect to $t'$, the one with minimum partial cost has partial cost 
\begin{align}
    \min\limits_{v \in B_{t}} \{ \mathcal{T}_{t'}[F,\mathcal{C} \cup \{\{u,v\}\}, \abov',\below'] \} \label{eqt:forget-node-1}
\end{align}
where $\abov'$ is the extension of $\abov$ to $B_{t'}$ with $\abov'(u)=0$ and $\below'$ is the extension of $\below$ to $B_{t'}$ with $\below'(u)=n_G-1$. 

\begin{proof}[Proof of case 1 -- \Cref{eqt:forget-node-1}]
    Let $T$ be a solution satisfying $(F,\mathcal{C}, \abov, \below)$, where there is no vertex above $u$ in $T$ with respect to $t'$, $u$ has no neighbor in $B_t$ and $u$ has a single below connection $\{u,v\}, v \in B_t$ with respect to $t'$.
    See \Cref{fig:tw-forget-node-1} for an illustration of this case.

\begin{figure}[t]
    \centering
    \begin{tikzpicture}
        \begin{scope}[on background layer]
            \draw[draw=none, fill=black!10,rounded corners=15pt] (0.375,-0.375) rectangle (5.625,1.25);
        \end{scope}
        \draw[loosely dashed] (0.375,0.25)--(5.625,0.25);
        \node[vertex_small] (u) at (4.5,-0.05) [label=right:{$u$}] {};
        \node[vertex_small] (v) at (5.25,0.875) [label=right:{$v$}] {};
        \node[vertex_small] (a) at (0.75,0.5) {};
        \node[vertex_small] (b) at (1.5,1) {};
        \node[vertex_small] (c) at (2.25,0.5) {};
        \node[vertex_small] (d) at (3,1) {};
        \node[vertex_small] (e) at (3.75,0.5) {};
        \node[vertex_small] (f) at (4.5,0.75) {};
        \draw (a)--(b) (e)--(f)--(v);
        \node[vertex_small, color=CornflowerBlue] (f1) at (0.75,-0.75) {};
        \node[vertex_small, color=CornflowerBlue] (f2) at (1.5,-0.75) {};
        \node[vertex_small, color=CornflowerBlue] (f3) at (2.25,-0.75) {};
        \node[vertex_small, color=CornflowerBlue] (f4) at (3,-0.75) {};
        \node[vertex_small, color=CornflowerBlue] (f5) at (3.75,-0.75) {};
        \node[vertex_small, color=CornflowerBlue] (f6) at (4.5,-0.75) {};
        \node[vertex_small, color=CornflowerBlue] (f7) at (5.25,-0.75) {};
        \node[vertex_small, color=CornflowerBlue] (f8) at (6,-0.75) {};
        \draw[CornflowerBlue] (a)--(f1) (a)--(f2) (c)--(f2) (e)--(f4) (u)--(f6) (v)--(f7) (v)--(f8) (f2)--(f3)--(f4) (f5)--(f6)--(f7);
        \node[vertex_small, color=Apricot] (h1) at (1.5,1.75) {};
        \node[vertex_small, color=Apricot] (h2) at (3,1.75) {};
        \node[vertex_small, color=Apricot] (h3) at (3.75,1.75) {};
        \node[vertex_small, color=Apricot] (h4) at (5.25,1.75) {};
        \draw[Apricot] (h1)--(b) (d)--(h2) (e)--(h3) (v)--(h4) (h2)--(h3);
        \draw [decorate, decoration = {brace}] (5.775,1.25) -- node [midway, right] {$B_{t}$} (5.775,0.125);
        \draw [decorate, decoration = {brace}] (6.375,1.25) -- node [midway, right] {$B_{t'}$} (6.375,-0.375);
        \begin{scope}[shift={(4.5,0.0625)}, scale=1.4, shift={(-4.5,-0.0625)}]
            \draw[dotted, rounded corners, xshift=-3pt, yshift=3pt] (5.25,0.875)--(5.25,-0.75)--(3.75,-0.75)--cycle;
        \end{scope}
    \end{tikzpicture}
    \caption{\textbf{Case 1:} $u$ is connected to $B_t$ in $T$ with a single below connection $\{u,v\}$. All vertices (except $u$) are below $u$ with respect to $t'$ and no vertex is above. No edge is forgotten when $u$ is forgotten.}
    \label{fig:tw-forget-node-1}
\end{figure}
    
    Since forgotten edges are identical with respect to $t$ and $t'$, the partial cost of $T$ with respect to $t$ is the same as its partial cost with respect to $t'$, below connections in $T$ (other than $\{u,v\}$) are identical with respect to $t$ and $t'$, and for every $w \in B_t$, the same vertices are below $w$ in $T$ with respect to $t$ and $t'$.
    
    Since hidden edges are identical with respect to $t$ and $t'$, for every $w \in B_t$ the same vertices are above $w$ in $T$ with respect to $t$ and $t'$.
    
    Since $\{u\}$ is a connected component in $T[B_{t'}]$ and there is no vertex above $u$ in $T$ with respect to $t'$, there must be $n_G-1$ vertices below $u$ in $T$ with respect to $t'$.
    Hence \Cref{eqt:forget-node-1}.
\end{proof}

\subparagraph{Case 2:} Among the solutions where $u$ is connected to a unique vertex $v \in B_t$ via an edge $uv$, the one with minimum partial cost has partial cost
\begin{align}
    \min\limits_{\substack{v \in N_G(u) \cap B_{t} \\ i \in \{0, \ldots, n_G-2\}}} \{ \mathcal{T}_{t'}[F \cup \{uv\},\mathcal{C}, \abov',\below'] + (i+1)(n_G - (i+1)) \} \label{eqt:forget-node-2}
\end{align}
where $\abov'$ is the extension of $\abov$ to $B_{t'}$ with $\abov'(u)=0$, and for each $v \in N_G(u) \cap B_{t}$, for each $i \in \{0, \ldots, n_G-2\}$, $\below'$ is given by
\begin{equation*}
    \below'(w)=
    \begin{cases}
      \below(w) & \text{if}\ w \neq u,v \\
      i & \text{if}\ w=u \\
      \below(v)-(i+1)\ & \text{if}\ w=v.
    \end{cases}
\end{equation*}

\begin{proof}[Proof of case 2 -- \Cref{eqt:forget-node-2}]
    Let $T$ be a solution satisfying $(F,\mathcal{C}, \abov, \below)$ where there is no vertex above $u$ in $T$ with respect to $t'$, where $u$ has a single neighbor $v \in B_t$ and $u$ appears in no below connection with respect to $t'$.
    See \Cref{fig:tw-forget-node-2} for an illustration of this case.
    
\begin{figure}[t]
    \centering
    \begin{tikzpicture}
        \begin{scope}[on background layer]
            \draw[draw=none, fill=black!10,rounded corners=15pt] (0.375,-0.375) rectangle (5.625,1.25);
        \end{scope}
        \draw[loosely dashed] (0.375,0.25)--(5.625,0.25);
        \node[vertex_small] (u) at (4.5,-0.05) [label=right:{$u$}] {};
        \node[vertex_small] (v) at (5.25,0.875) [label=right:{$v$}] {};
        \node[vertex_small] (a) at (0.75,0.5) {};
        \node[vertex_small] (b) at (1.5,1) {};
        \node[vertex_small] (c) at (2.25,0.5) {};
        \node[vertex_small] (d) at (3,1) {};
        \node[vertex_small] (e) at (3.75,0.5) {};
        \node[vertex_small] (f) at (4.5,0.5) {};
        \draw (a)--(b) (e)--(f)--(v);
        \draw[CornflowerBlue, thick] (u)--(v);
        \node[vertex_small, color=CornflowerBlue] (f1) at (0.75,-0.75) {};
        \node[vertex_small, color=CornflowerBlue] (f2) at (1.5,-0.75) {};
        \node[vertex_small, color=CornflowerBlue] (f3) at (2.25,-0.75) {};
        \node[vertex_small, color=CornflowerBlue] (f4) at (3,-0.75) {};
        \node[vertex_small, color=CornflowerBlue] (f5) at (3.75,-0.75) {};
        \node[vertex_small, color=CornflowerBlue] (f6) at (4.5,-0.75) {};
        \node[vertex_small, color=CornflowerBlue] (f7) at (5.25,-0.75) {};
        \node[vertex_small, color=CornflowerBlue] (f8) at (6,-0.75) {};
        \draw[CornflowerBlue] (a)--(f1) (a)--(f2) (c)--(f2) (e)--(f4) (u)--(f6) (v)--(f7) (v)--(f8) (f2)--(f3)--(f4) (f5)--(f6);
        \node[vertex_small, color=Apricot] (h1) at (1.5,1.75) {};
        \node[vertex_small, color=Apricot] (h2) at (3,1.75) {};
        \node[vertex_small, color=Apricot] (h3) at (3.75,1.75) {};
        \node[vertex_small, color=Apricot] (h4) at (5.25,1.75) {};
        \draw[Apricot] (h1)--(b) (d)--(h2) (e)--(h3) (v)--(h4) (h2)--(h3);
        \draw [decorate, decoration = {brace}] (5.775,1.25) -- node [midway, right] {$B_{t}$} (5.775,0.125);
        \draw [decorate, decoration = {brace}] (6.375,1.25) -- node [midway, right] {$B_{t'}$} (6.375,-0.375);
        \begin{scope}[on background layer]
        \draw[line width=6.5pt, color=lightgray,cap=round] (u)--(v);
        \end{scope}
    \end{tikzpicture}
    \caption{\textbf{Case 2:} $u$ is connected to $B_t$ with a single edge in $T$ (highlighted in grey). There are 2 vertices below $u$ with respect to $t'$, so there are 5 vertices below $v$ with respect to $t$ (2 vertices are below $v$ with respect to $t$ and $t'$, plus the 2 vertices below $u$ and $u$ itself). The cost of $uv$ is $3(n_G-3)$ since there are 2 vertices below $u$ when it is forgotten.}
    \label{fig:tw-forget-node-2}
\end{figure}

    Only the edge $uv$ is forgotten when $u$ is forgotten, so the partial cost of $T$ with respect to $t$ is its partial cost with respect to $t'$ plus the cost of the edge $uv$.
    
    The edges of the forest $T[B_{t'}]$ are $F \cup \{uv\}$.
    Below connections are identical with respect to $t$ and $t'$.
    
    Hidden edges are identical with respect to $t$ and $t'$ so for every $w \in B_t$, the same vertices are above $w$ in $T$ with respect to $t$ and $t'$.
    
    For any $w \in B_{t} \setminus \{v\}$, no edge incident to $w$ is forgotten when $u$ is forgotten, so the same vertices are below $w$ in $T$ with respect to $t$ and $t'$.
    For $v$ however, the vertices below $u$ in $T$ with respect to $t'$ and $u$ itself are below $v$ with respect to $t$. 
    If there are $i$ vertices below $u$ with respect to $t'$, then there are $\below(v)-(i+1)$ vertices below $v$ in $T$ with respect to $t'$.
    
    The cost of the edge $uv$ is $(i+1)(n_G - (i+1))$ because there are $i+1$ vertices closer to $u$ than to $v$ in $T$.

    We keep the best solution among all possible neighbors $v \in B_t \cap N_G(u)$ and all possible values of $\below'(u)=i \in \{0, \ldots, n_G-2\}$.
    Hence \Cref{eqt:forget-node-2}.
\end{proof}

\subparagraph{Case 3:} Among the solutions where $u$ has several edges or below connections to vertices of $B_t$ with respect to $t'$, the one with minimum partial cost has partial cost 
\begin{align}
    \min\limits_{\substack{C \in \mathcal{C} \\ X \subseteq C \cap N_G(u) \\ p \in \mathcal{P}(C \setminus X) \\ 0 \leq \below'' \leq \below}} \{ \mathcal{T}_{t'}[F \cup \{uv \colon v \in X\},(\mathcal{C} \setminus \{C\}) \cup \{Y \cup \{u\} \colon Y \in p \}, \abov',\below'] + \omega \} \label{eqt:forget-node-3}
\end{align}
where $\mathcal{P}(C \setminus X)$ is the set of partitions of $C \setminus X$, $\below''\colon X \cup \{u\} \to \{0, \ldots, n_G-1 \}$ verifies $0 \leq \below''(v) \leq \below(v)$ for all $v \in X$ and $\below''(u) \in \{0, \ldots, n_G-1-|X|\}$, $\abov'$ is the extension of $\abov$ to $B_{t'}$ with $\abov'(u)=0$, and $\below'$ is given by:
\begin{equation*}
    \below'(v)=
    \begin{cases}
      \below(v) & \text{if}\ v \notin X \cup \{u\} \\
      \below''(v) & \text{if}\ v \in X \cup \{u\}
    \end{cases}
\end{equation*}
and $\omega$ is the cost of all edges between $u$ and $X$. 
Let $v \in X$ and let $F^{v}$ be the connected component containing $v$ in the forest induced by $F$.
The vertices closer to $v$ than to $u$ are those that are in $F^{v}$ or above or below a vertex of $F^{v}$, so their number is given by:
\[S_v = \sum\limits_{x \in V(F^{v})} \abov'(x)+\below'(x)+1.\]
Thus the cost of the edge $uv$ is given by $S_v \left(n_G-S_v \right)$ and the cost of all edges forgotten between $t'$ and $t$ is
\[\omega= \sum\limits_{v \in X} S_v \left(n_G-S_v \right). \]

\begin{proof}[Proof of case 3 -- \Cref{eqt:forget-node-3}]
    Let $T$ be a solution satisfying $(F,\mathcal{C}, \abov, \below)$ where there is no vertex above $u$ in $T$ with respect to $t'$ and such that $u$ has several edges or below connections with vertices of $B_t$ with respect to $t'$.
    See \Cref{fig:tw-forget-node-3} for an illustration of this case.
    
\begin{figure}[t]
    \centering
    \begin{tikzpicture}[label distance=-2pt, xscale=1.2]
        \begin{scope}[on background layer]
            \draw[draw=none, fill=black!10,rounded corners=15pt] (0.375,-0.375) rectangle (5.625,1.5);
        \end{scope}
        \draw[loosely dashed] (0.375,0.25)--(5.625,0.25);
        \node[vertex_small] (u) at (3,-0.05) [label=right:{$u$}] {};
        \node[vertex_small] (v) at (0.75,0.6) [label=left:{$v$}] {};
        \node[vertex_small] (w) at (1.5,0.6) [label=left:{$w$}] {};
        \node[vertex_small] (x) at (2.25,0.6) [label=above:{$x$}] {};
        \node[vertex_small, draw=black, fill=white] (y) at (3,1) [label=left:{$y$}] {};
        \node[vertex_small, draw=black, fill=white] (z) at (3.75,0.6) [label=above left:{$z$}] {};
        \node[vertex_small] (a) at (1.5,1.25) {};
        \node[vertex_small] (b) at (4.5,0.5) {};
        \node[vertex_small] (c) at (4.5,1.25) {};
        \node[vertex_small] (d) at (5.25,0.5) {};
        \node[vertex_small] (e) at (5.25,1.25) {};
        \draw (a)--(v) (z)--(b)--(c) (d)--(e);
        \draw[CornflowerBlue] (y)--(u) (z)--(u);
        \node[vertex_small, color=CornflowerBlue] (f1) at (1.125,-0.75) {};
        \node[vertex_small, color=CornflowerBlue] (f2) at (2.25,-0.6) {};
        \node[vertex_small, color=CornflowerBlue] (f3) at (3,-0.85) {};
        \node[vertex_small, color=CornflowerBlue] (f4) at (3.75,-0.75) {};
        \node[vertex_small, color=CornflowerBlue] (f5) at (4.5,-0.75) {};
        \draw[CornflowerBlue] (v)--(f1) (w)--(f1) (x)--(f2) (u)--(f2) (u)--(f3) (u)--(f4) (z)--(f5) (f1)to[bend right=15](f3);
        \node[vertex_small, color=Apricot] (h1) at (0.75,1.75) {};
        \node[vertex_small, color=Apricot] (h2) at (3,1.75) {};
        \node[vertex_small, color=Apricot] (h3) at (4.5,1.75) {};
        \node[vertex_small, color=Apricot] (h4) at (5.25,1.75) {};
        \draw[Apricot] (h1)--(v) (y)--(h2) (z)--(h3) (e)--(h4) (h3)--(h4);
        \draw [decorate, decoration = {brace}] (5.775,1.25) --  (5.775,0.125);
        \node (l1) at (5.7,0.625) [label=right:{$B_{t}$}] {};
        \draw [decorate, decoration = {brace}] (6.375,1.25) --  (6.375,-0.375);
        \node (l2) at (6.3,0.375) [label=right:{$B_{t'}$}] {};
        \node[draw, ellipse, rotate fit=-28, fit=(y)(z), inner xsep=2pt, inner ysep=4pt, dotted, label=above:{$X$}] {};
        \begin{scope}[on background layer]
        \draw[line width=6pt, color=black!20,cap=round] (z)--(u.center) (y)--(u.center) (u.center)--(f4) (u.center)--(f3) (u.center)--(f2) (f2)--(x) (f1)to[bend right=15](f3) (f1)--(w) (f1)--(v); 
        \end{scope}
    \end{tikzpicture}
    \caption{\textbf{Case 3:} $u$ is connected (with edges and/or below connections) to several vertices in $B_t$ that become a single below connection when $u$ is forgotten. There is a below connection $C$ with respect to $t$, highlighted in grey, whose associated subtree contains $u$. With respect to $t'$, it is divided into $X=\{y,z\}$ the neighbors of $u$ in $B_t$ (in white) and $C \setminus X=\{v,w,x\}$ the other vertices in the below connection. These vertices of $C \setminus X$ are partitioned into $\{v,w\}, \{x\}$ and each set of the partition gives a below connection with $u$ with respect to $t'$, here $\{v,w,u\}$ and $\{x,u\}$. The edges between $u$ and $X$ are forgotten when $u$ is forgotten.}
    \label{fig:tw-forget-node-3}
\end{figure}

    Denote $C \in \mathcal{C}$ the below connection with respect to $t$ whose associated subtree contains $u$.
    Denote $X=N_T(u) \cap B_t$ the neighbors of $u$ in $B_t$. 
    The vertices of $X$ are leaves in the subtree associated with $C$, so $X \subseteq N_G(u) \cap C$.
    The partial cost of $T$ with respect to $t$ is its partial cost with respect to $t'$ plus the cost of the edges between $u$ and $X$.
    The edges of the forest $T[B_{t'}]$ are $F \cup \{uv \colon v \in X \}$.
    The vertices of $C \setminus X$ are all involved in a below connection with $u$ with respect to $t'$.
    Denote $Y_1, \ldots, Y_\ell$ the below connections involving $u$ in $T$ with respect to $t'$.
    Then $(Y_1 \setminus \{u\}, \ldots, Y_\ell \setminus \{u\})$ is a partition of $C \setminus X$.
    The hidden edges in $T$ are identical with respect to $t$ and $t'$, so for every $w \in B_t$, the vertices above $w$ in $T$ are identical with respect to $t$ and $t'$.
    Let $w \in B_t$. 
    If $w \notin X$, the forgotten edges incident to $w$ are identical with respect to $t$ and $t'$, so the same vertices are below $w$ in $T$ with respect to $t$ and $t'$.
    If $v \in X$, then since the edge $uv$ is forgotten in $t$ but not in $t'$, the number of vertices below $v$ in $T$ with respect to $t'$ is in $\{0, \ldots, \below(v)\}$.
    The number of vertices closer to $v$ than to $u$ in $T$ can be counted with the vertices in the connected component containing $v$ in the forest $T[B_t]$, denoted $T[B_t]^{v}$:
    \[S_v = \sum\limits_{x \in V(T[B_t]^{v})} \abov'(x)+\below'(x)+1\]
    so the cost of that edge is $S_v (n_G - S_v)$.
\end{proof}

Among all the possible below connections $C \in \mathcal{C}$ whose associated subtree contains $u$, all the possible sets of neighbors of $u$ $X \subseteq C \cap N_G(u)$, all the possible combinations of below connections between the vertices of $C \setminus X$ $p \in \mathcal{P}(C \setminus X)$ and all the possible values for $\below''$, we keep the best solution.

Among the three sets of suitable indices \eqref{eqt:forget-node-1}, \eqref{eqt:forget-node-2}, \eqref{eqt:forget-node-3}, we store the best partial cost in~$\mathcal{T}_t[F,\mathcal{C}, \abov, \below]$.

\paragraph*{Running time.}
Let $k$ be the treewidth of the input graph $G$. 
We can compute in time $2^{O(k^3)} n_G$ a tree decomposition of $G$ of width $k$ and transform it in time $O(k^2 n_G)$ into a nice tree decomposition of width $k$ and with at most $O(k n_G)$ nodes \cite{CFKLMPPS15}. 
The time required to compute the tree decomposition is negligible compared to that of the dynamic programming algorithm, so we ignore it in the following.

The size of a table $\mathcal{T}_t$ is at most $2^{2^{k+1}+k(k+1)}n_G^{2(k+1)}$.
We can ignore leaves and introduce nodes in the running time, as the time of computing an entry of the table in these cases is dominated by the other cases.
For join nodes, computing an entry of the table takes at most $n_G^{k+1}2^{2^{k+1}}$ time.
For forget nodes, we can focus on case 3 which is the most computationally demanding. Computing an entry of the table in case 3 takes at most $4^{k+1}(k+1)^{k+3}n_G^{k+1}$ time (including computing the cost of forgotten edges).
Overall, processing a node takes time at most $2^{O(2^k)} n_G^{O(k)}$.

With at most $O(k n_G)$ nodes to process, the dynamic programming algorithm runs in time $n_G^{O(k)}2^{O(2^k)} k$.
This concludes the proof of \cref{thm:treewidth-XP}.
\end{proof}

It remains an open question whether \MADST is FPT with respect to treewidth.

\subsection{Modular Width}
Although less prominent than treewidth, modular width is a well-established width parameter.
An advantage over treewidth is that modular width can be computed in~$O(n+m)$ time~\cite{TCHP08}.
As usual, modular width comes with a decomposition encoding the input graph and algorithms typically process this decomposition in a bottom-up manner while maintaining partial solutions for the problem at hand~\cite{GLO13,KN18,CDP19}.
Our approach breaks with this standard approach: we only use a small part of the structure provided by modular width in our branching algorithm as we can show that \MADs have a very special structure in graphs of low modular width.

Before describing the high-level idea of our algorithm, we need to introduce some basics on modular width and modular decompositions first.
We follow the notation of \citet{KN18}.
A \emph{module} of a graph~$G$ is a set of vertices~$M \subseteq V(G)$, such that for any vertex~$x \in V(G)\setminus M$, either~$M \subseteq N(x)$ or~$M\cap N(x) = \emptyset$, meaning all vertices of~$M$ share the same neighborhood in~$V(G)\setminus M$.
A partition~$P = \{M_1, M_2, \dots, M_{\ell}\}$ of the vertices~$V(G)$ into~$\ell \geq 2$ modules of~$G$ is called a \emph{modular partition}.
If there are~$v \in M_i$ and~$u \in M_j$ with~$uv \in E(G)$, then all vertices in~$M_i$ are adjacent to all vertices in~$M_j$ and we say modules~$M_i$ and~$M_j$ are \emph{adjacent}.
The \emph{quotient graph}~$G_{/P}$ of~$G$ has vertices~$\{q_{1}, q_{2}, \dots, q_{\ell}\}$ and two vertices~$q_{i}$ and~$q_{j}$ are adjacent if and only if the modules~$M_i$ and~$M_j$ are adjacent.

Note that the quotient graph~$G_{/P}$ only stores the information whether or not two vertices~$u,v \in V(G)$ from \emph{different} modules are adjacent.
What is missing is the information on edges within modules; this is typically covered recursively: for a module~$M_i$ one considers a modular partition of~$G[M_i]$ until the modules are single vertices.
This leads to the \emph{modular decomposition} and the modular width is the largest number of nodes in any employed quotient graph in the decomposition; we refer to \citet{GLO13} for a formal definition.

Our approach to find a \MAD only needs a modular partition~$P$ of the input graph~$G$ with the parameter~$k$ being the number of modules in~$P$.
Thus, $k$ is upper bounded by the modular width but can be much smaller as we only need a single modular partition~$P$ of~$G$ with a minimum number of modules and its corresponding quotient graph~$G_{/P}$.
For an intuition why considering~$G_{/P}$ suffices, recall that the $n$-vertex trees with minimum and maximum Wiener index are the star~$K_{1,n-1}$ and the path~$P_n$~\cite{DEG01}.
Intuitively, a \MAD should be as close as possible to a star.
This is where the modularity is beneficial: a vertex~$v$ is either adjacent to all or none of the vertices in another module. 
Thus, for two adjacent modules~$M_1$ and~$M_2$, a \MAD can be very ``star-like''.
More precisely, if the modular partition consists of only two modules~$M_1$ and~$M_2$, then it is not too difficult to verify that an optimal solution is either a star or what we call a ``double-star''---a tree with exactly two inner vertices (see \Cref{fig:polystar}).
\begin{figure}[t]
    \centering
    \begin{tikzpicture}
        \begin{scope}
        \node[vertex_small, color=CornflowerBlue] (a1) at (1,4) {};
        \node[vertex_small] (b1) at (1,4.5) {};
        \node[vertex_small] (c1) at (0.5,3.5) {};
        \node[vertex_small] (d1) at (0.75,3) {};
        \node[vertex_small] (e1) at (0.5,2.5) {};
        \node[vertex_small] (f1) at (1,2) {};
        \foreach \x [count=\i] in {b1,c1,d1}{
            \draw (\x) -- (a1);
        }
        \draw (d1)--(e1)--(f1);
        \node[vertex_small] (g1) at (2.5,4) {};
        \node[vertex_small,color=CornflowerBlue] (h1) at (2.5,3) {};
        \node[vertex_small] (i1) at (2.5,2) {};
        \draw (g1)--(h1)--(i1);
        \foreach \i in {a1,b1,c1,d1,e1,f1}
        {
            \foreach \j in {g1,h1,i1}{
                \draw (\i) -- (\j);
            }
        }
    \end{scope}
    \begin{scope}[xshift=4.5cm]
        \node[vertex_small, color=CornflowerBlue] (a2) at (1,4) {};
        \node[vertex_small] (b2) at (1,4.5) {};
        \node[vertex_small] (c2) at (0.5,3.5) {};
        \node[vertex_small] (d2) at (0.75,3) {};
        \node[vertex_small] (e2) at (0.5,2.5) {};
        \node[vertex_small] (f2) at (1,2) {};
        \foreach \x [count=\i] in {b2,c2,d2}{
            \draw (\x) -- (a2);
        }
        \draw (d2)--(e2)--(f2);
        \node[vertex_small] (g2) at (2.5,4) {};
        \node[vertex_small,color=CornflowerBlue] (h2) at (2.5,3) {};
        \node[vertex_small] (i2) at (2.5,2) {};
        \draw (g2)--(h2)--(i2);
    \begin{scope}[on background layer]
        \node[draw=none, fill=black!10, rounded corners, fit=(a2)(b2)(c2)(d2)(e2)(f2), label=left:$M_1$] (M1) {};
        \node[draw=none, fill=black!10, rounded corners, fit=(g2)(h2)(i2), label=right:$M_2$] (M2){};
        \draw[dashed, thick, gray] (M1)--(M2);
    \end{scope}
    \end{scope}

    \begin{scope}[xshift=9cm]
        \node[vertex_small, color=CornflowerBlue] (a3) at (1,4) {};
        \node[vertex_small] (b3) at (1,4.5) {};
        \node[vertex_small] (c3) at (0.5,3.5) {};
        \node[vertex_small] (d3) at (0.75,3) {};
        \node[vertex_small] (e3) at (0.5,2.5) {};
        \node[vertex_small] (f3) at (1,2) {};
        \node[vertex_small] (g3) at (2.5,4) {};
        \node[vertex_small,color=CornflowerBlue] (h3) at (2.5,3) {};
        \node[vertex_small] (i3) at (2.5,2) {};

        \begin{scope}[on background layer]
            \node[draw=none, fill=black!10, rounded corners, fit=(a3)(b3)(c3)(d3)(e3)(f3), label=left:$M_1$] (M1) {};
            \node[draw=none, fill=black!10, rounded corners, fit=(g3)(h3)(i3), label=right:$M_2$] (M2) {};
        \end{scope}
        \foreach \i/\j in {a3/g3, a3/h3, a3/i3, e3/h3, f3/h3}{
            \draw (\i) -- (\j);
        }
        \draw (a3) -- (b3) (a3) -- (c3) (a3)--(d3);
    \end{scope}

    \begin{scope}[scale=1.3]
        \begin{scope}[xshift=1.5cm, yshift=-4.5cm]
            \node[vertex_small, color=CornflowerBlue] (1a) at (0.375,4.5) {};
            \node[vertex_small] (1b) at (0.375,5) {};

            \node[vertex_small] (2a) at (0.375,2.5) {};
            \node[vertex_small, color=CornflowerBlue] (2b) at (0.375,3) {};
            
            \draw[densely dotted] (2a)--(2b);

            \node[vertex_small] (3a) at (2.125,3.25) {};
            \node[vertex_small, color=CornflowerBlue] (3b) at (2.125,3.75) {};
            \node[vertex_small] (3c) at (2.125,4.25) {};
            \draw[densely dotted] (3a)--(3b)--(3c);

            \node[vertex_small, color=Apricot] (4r) at (3.75,4.25) {};
            \node[vertex_small] (4a) at (3.5,4.75) {};
            \node[vertex_small] (4b) at (4,4.75) {};
            \node[vertex_small] (4c) at (3.5,3.75) {};
            \node[vertex_small] (4d) at (4,3.75) {};
            \node(root) at (3.75,5.2) {Root module};
            \draw[densely dotted] (4r)--(4a) (4r)--(4b) (4a)--(4b) (4r)--(4d);

            \node[vertex_small] (5a) at (6.5,4) {};
            \node[vertex_small, color=CornflowerBlue] (5b) at (6.5,4.5) {};
            \node[vertex_small] (5c) at (6.5,5) {};
            \draw[densely dotted] (5a)--(5b);

            \node[vertex_small] (6a) at (5.625,2.5) {};
            \node[vertex_small, color=CornflowerBlue] (6b) at (5.625,3) {};

            \draw[thick] (1a)--(3b) (1b)--(3b) (2a)--(3b) (2b)--(3b) (3a)--(4r) (3b)--(4r) (3c)--(4r) (4a)--(4r) (4b)--(4r) (4c)--(3b) (5a)--(4r) (5b)--(4r) (5c)--(4r) (6a)--(4r) (6b)--(4r) (4r)--(4d);

            \begin{scope}[on background layer]
            \node[draw=none, fill=black!10, rounded corners,fit=(1a) (1b),label=left:$M_1$](m1) {};
            \node[draw=none, fill=black!10, rounded corners,fit=(2a) (2b),label=left:$M_2$](m2) {};
            \node[draw=none, fill=black!10, rounded corners,fit=(3a) (3b)(3c),label=above:$M_3$](m3) {};
            \node[draw=Apricot, very thick, rounded corners,fit=(4a) (4b)(4c)(4r)(4d),label=below:$M_4$](m4) {};
            \node[draw=none, fill=black!10, rounded corners,fit=(5a) (5b)(5c),label=right:$M_5$](m5) {};
            \node[draw=none, fill=black!10, rounded corners,fit=(6a) (6b),label=right:$M_6$](m6) {};
            \draw[thick, dashed, gray] (m1)--(m3) (m2)--(m3) (m3.40)--(m4.160) (m4.40)--(m5.160) (m5.270)--(m6.30) (m6)--(m4) (m6) -- (m2);
            \end{scope}
        \end{scope}
        \end{scope}
    \end{tikzpicture}
    \caption{The three top figures represent a case with two modules: on the left is the input graph, in the middle is the quotient graph (in grey dashed edges) and the internal edges of each module, and on the right is a double-star with centers highlighted in blue (not necessarily a \MAD).
    The bottom figure represents a poly-star in solid edges (not necessarily a \MAD).
    The root and the root module are highlighted in orange.
    One maximum-degree vertex per module is highlighted in blue. Internal edges unused in the poly-star are dotted.
    }
    \label{fig:polystar}
\end{figure}
The internal edges within~$G[M_1]$ and~$G[M_2]$ only influence which vertices can be inner vertices of the double-star (namely the vertices with highest degree in each module).
This implies that a \MAD can be found without access to optimal sub-solutions of~$M_1$ and~$M_2$.
It is therefore sufficient to only consider a single quotient graph~$G_{/P}$ of~$G$.

We generalize this approach to more than two modules by showing that we can always find a \MAD which is a \emph{poly-star}, that is, a tree in which there is at most one non-leaf in any module.
We restrict the structure of an optimal solution even further by showing several properties:
Crucially, there is one dedicated module, called the \emph{root module}.
For one of the maximum-degree vertices in the root module (called the \emph{root}) the entire neighborhood (in the input graph) is contained in the poly-star.
Beside these edges, the poly-star does not contain any edges between vertices that are in the same module in the input graph.
Furthermore, we show that all vertices of a non-root module are adjacent to a single vertex in the adjacent module that is closest to the root module in the quotient graph.
All vertices in the root module that are not adjacent to the root are adjacent to a single vertex in an adjacent module.
An example of such a poly-star can be seen in \Cref{fig:polystar}.

Our algorithm is given a modular partition (which is computable in linear time~\cite{TCHP08}) and finds such a poly-star as follows: 
First, enumerate all possible spanning trees of the quotient graph and each option for the root module.
For each option, create a poly-star satisfying the additional criteria.
Finally, choose the best poly-star among all considered options.
Our algorithm is shown in \Cref{alg:mw}.
\begin{algorithm}[t]
 \caption{Solving \MADST in FPT time with respect to modular width}
  \label{alg:mw}
 \newdimen\inwidth
 \setbox0=\hbox{\hspace*{\algorithmicindent}\textbf{Output:}}
 \inwidth=\wd0
 \hspace*{\algorithmicindent}\textbf{Input:} \begin{minipage}[t]{\linewidth-\inwidth}
  A connected graph~$G$ with modular partition $P=\{M_1,\ldots,M_k\}$ and quotient graph~$G_{/P}$ with vertices~$\{q_1,\dots,q_k \}$; budget~$b \in \N$.
 \end{minipage}
 \hspace*{\algorithmicindent}\textbf{Output:}
 \begin{minipage}[t]{\linewidth-\inwidth}
  \texttt{TRUE} if there exists a \MAD~$T$ of~$G$ with~$\W(T)\leq b$ and \texttt{FALSE} otherwise.
 \end{minipage}
 \begin{algorithmic}[1]

  \State For a vertex~$v \in V(G)$, define~$\module(v) \coloneqq i\in[k]$, where~$v\in M_i$.
  \State opt $\coloneqq \infty$
  \ForEach {$i \in [k]$}\label{alg_mw:i_loop}
    \State Choose a vertex~$r \in M_i$ with maximum degree in $G$ as root
    \ForEach {spanning tree~$T'$ of~$G_{/P}$ with~$N_{T'}(q_i)=N_{G_{/P}}(q_i)$}\label{alg_mw:t_loop}
      \State Initialize~$T\coloneqq (V(G),\emptyset)$
      \State Choose an arbitrary vertex~$d_j \in M_j$ for every~$j \in [k]\setminus\{i\}$ and set~$d_i\coloneqq r$
      \State Let~$D\coloneq\{d_{\ell}\mid \ell \in [k]\}$
      \State For each edge~$q_{\ell}q_{\ell'}$ in~$E(T')$, add~$d_{\ell}d_{\ell'}$ to~$E(T)$ 
      \ForEach {module~$M_j$ with~$j\neq i$}\label{alg_mw:cn_a1}
        \State Let~$c_j\in D$ be such that~$q_{\module(c_j)}$ is the unique vertex in~$P^{T'}_{q_i,q_j}\cap N_{T'}(q_j)$
        \State For every~$u \in M_j\setminus \{d_j\}$, add the edge~$uc_j$ to~$E(T)$
      \EndFor\label{alg_mw:cn_a2}
      \State For every~$u \in N_{G}(r)\cap M_i$, add the edge~$ur$ to~$E(T)$\label{alg_mw:root_neighbors}
      \State Let~$R \coloneqq T - (M_i\setminus N_T[r])$
       and choose some~$c_i\in\argmin\limits_{w \in N_T(r)\cap D}\dist_R(w)$ \label{alg_mw:cn_b1}
      \State For every~$u \in M_i\setminus N_G[r]$, add the edge~$uc_i$ to~$E(T)$\label{alg_mw:cn_b2}
      \State opt $\coloneqq \min(\text{opt},\W(T))$
    \EndFor
  \EndFor
  \State\Return opt~$\leq b$
 \end{algorithmic}
\end{algorithm}

\begin{theorem}
	\label{thrm:mw_fpt}
	\MADST can be solved in~$O(k^{k+1}(m+n))$ time where~$k$ is the minimum number of modules in any modular partition of the input graph~$G$.
\end{theorem}

For the correctness of \Cref{alg:mw}, we first prove a general lemma which does not rely on a modular partition and which describes the change in the Wiener index of a tree when repositioning a family of subtrees.
\Cref{obs:common_anchor} is visualized in \Cref{fig:lem_common_anchor}.
Using this and, crucially, \Cref{lem:Dankelmann} we then prove the structural observations we informally stated above relying on a modular partition.

\begin{figure}[t]
    \centering
    \begin{tikzpicture}[xscale=0.75, yscale=0.5]
        \begin{scope}
        \begin{scope}[on background layer]
            \draw[draw=none, fill=black!10, rounded corners] (0, -0.5) rectangle (6, 3);
        \end{scope}
        \node[vertex_small, label=left:$a$] (a) at (1.5,2) {};
        \node[vertex_small, label=right:$b$] (b) at (4.5,2) {};
        \node (Tr) at (-0.75,1.25) {$T_R$};
        \draw (a) to[out=-30, in=210] (2.5,0.5)
                            to[out=30,  in=180] (3,1.5)
                            to[out=0,  in=180] (3.5,1)
                            to[out=0, in=-90] (b);
        \node(P) at (1.5,0.5) {$P^T_{a,b}$};

        \begin{scope}[on background layer]
            \draw[draw=none, fill=black!10, rounded corners] (0, 3.4) rectangle (1.7, 6.4);
        \end{scope}
        \node[vertex_small] (a1) at (0.85,4) {};
        \filldraw[draw=black, fill=white] (a1.north) -- (0.35,6) -- (1.35,6) -- cycle;
        \node[] (Tr) at (-1.95,4.9) {$T_A\coloneq\bigcup_{v \in V_a}T_{av}^v$};

        \begin{scope}[on background layer]
            \draw[draw=none, fill=black!10, rounded corners] (2.3, 3.4) rectangle (6, 6.4);
        \end{scope}
        \node[vertex_small] (b1) at (3.15,4) {};
        \node[vertex_small] (b2) at (5.15,4) {};
        \filldraw[draw=black, fill=white] (b1.north) -- (2.65,6) -- (3.65,6) -- cycle;
        \filldraw[draw=black, fill=white] (b2.north) -- (4.65,6) -- (5.65,6) -- cycle;
        \node (Tr) at (7.9,4.9) {$T_B\coloneq\bigcup_{v \in V_b}T_{bv}^v$};

        \draw (a)--(a1) (b)--(b1) (b)--(b2);
        \draw[thick, italyGreen, dashed] (a)--(b1) (a)--(b2);
        \draw[dashed] (b)--(a1);
        \end{scope}

    \end{tikzpicture}
    \caption{A sketch of \Cref{obs:common_anchor}.
    Given two options~$a$ and~$b$ to connect a family of subtrees of a spanning tree~$T$, it is always best to choose the same option based on the distance to the rest $T_R$ of the tree.
    The figure shows the case~$\dist_{T_R} (a) \leq \dist_{T_R} (b)$.
    This means it would be better to replace the two edges from~$b$ to~$T_B$ by the two dashed green edges.\\
    }
    \label{fig:lem_common_anchor}
\end{figure}

\begin{restatable}{lemma}{commonanchor}
    \label{obs:common_anchor}
    Let~$G$ be a graph with two distinct vertices~$a,b$ and two nonempty vertex sets~$V_a, V_b$ such that $\{a,b\} \subseteq N_G(v)$ for every~$v \in V_a \cup V_b$.
    Let~$T$ be a spanning tree of~$G$ such that~$V_a \subseteq N_T(a)$,~$V_b \subseteq N_T(b)$ and~$V(P^T_{a,b}) \cap (V_a \cup V_b) = \emptyset$.
    Let~$T_1$ and~$T_2$ be obtained from~$T$ by setting~$V(T_1) = V(T_2)\coloneq V(T)$ and
    \begin{align*}E(T_1)&\coloneq (E(T) \setminus \{bv\mid v \in V_b\}) \cup \{av\mid v \in V_b\},\\
    E(T_2)&\coloneq (E(T) \setminus \{av\mid v \in V_a\}) \cup \{bv\mid v \in V_a\}.
    \end{align*}
    Then,~$T_1$ and~$T_2$ are spanning trees of~$G$ and~$\min(\W(T_1),\W(T_2)) < \W(T)$.
\end{restatable}

\begin{proof}
    Let~$T_A\coloneq\bigcup_{v \in V_a}T_{av}^v$ and~$T_B\coloneq\bigcup_{v \in V_b}T_{bv}^v$ and let~$T_R \coloneqq T - T_A - T_B$.
    Intuitively,~$T_A$ are all subtrees attached to~$a$,~$T_B$ are all subtrees attached to~$b$ and~$T_R$ is ``the rest of the tree'' so~$T$ without both forests.
    Since~$T$ is a tree and~$V(P^T_{a,b}) \cap (V_a \cup V_b) = \emptyset$, it holds that~$V(T_A) \cap V(T_B) = \emptyset$ and~$P^T_{a,b} \subseteq T_R$.
    Without loss of generality, let~$\dist_{T_R}(a) \leq \dist_{T_R}(b)$ and set~$T' = T_1$.
    The case~$\dist_{T_R}(b) \leq \dist_{T_R}(a)$ is entirely analogous with~$T'=T_2$.

    We first show that~$T'$ is a spanning tree.
    For each~$v \in V_b$, removing the edge~$bv$ leaves two connected components~$T_{bv}^v$ and~$T_{bv}^b$ (which are trees) since~$T$ is a tree.
    By definition, $v \in V(T_{bv}^v)$ and~$a \in V(T_{bv}^b)$. Therefore, adding the edge~$av$ creates a spanning tree again.

    We now consider the difference in the Wiener index of~$T$ and~$T'$.
    Let~$A \coloneqq V(T_A), B\coloneqq V(T_B)$ and~$R\coloneqq V(T_R)$.
    Since~$T_R \cup T_A$ is a subtree of both~$T$ and~$T'$, the distances between vertices in~$R \cup A$ are identical in~$T$ and~$T'$ and we obtain
    \begin{align*}
    \W(T) - \W(T') = \;&\dist_T(B, B) + \dist_T(B,R) + \dist_T(B,A)\\
    &-(\dist_{T'}(B,B) + \dist_{T'}(B,R) + \dist_{T'}(B,A)).
    \end{align*}
    Note that~$\dist_T(B,B)=\dist_{T'}(B,B)$.
    Further, it holds that for all~$u \in B: \dist_ {T'}(u,a)=\dist_T(u,b)$ and for all~$v\in R\cup A$: $\dist_{T'}(a,v)=\dist_T(a,v)$.
    Since by assumption~$\dist_{T_R}(a)\leq\dist_{T_R}(b)$, it follows that
    \begin{align*}
    \dist_{T'}(B,R) &= \sum\limits_{u \in B}\sum\limits_{v \in R} \dist_{T'} (u,a) + \dist_{T'}(a,v) = \sum\limits_{u \in B}\sum\limits_{v \in R} \dist_{T} (u,b) + \dist_{T}(a,v)\\
    &\leq \sum\limits_{u \in B}\sum\limits_{v \in R} \dist_{T} (u,b) + \dist_{T}(b,v) = \dist_{T}(B,R).
    \end{align*}
    Since~$a\neq b$ and therefore~$\dist_T(a,b) > 0$, we further obtain
    \begin{align*}
    \dist_{T'}(B,A) &= \sum\limits_{u \in B}\sum\limits_{v \in A}\dist_{T'} (u,a)+\dist_{T'}(a,v) = \sum\limits_{u \in B}\sum\limits_{v \in A}\dist_T(u,b)+\dist_T(a,v)\\
    &<\sum\limits_{u \in B}\sum\limits_{v \in A}\dist_T(u,b)+\dist_T(b,a)+\dist_T(a,v)=\dist_T(B,A).
    \end{align*}
    Overall, we get~$\W(T)-\W(T')>0$ and therefore~$\W(T)>\W(T')$.
\end{proof}

We call a vertex~$r \in V(G)$ the \emph{root} of a MAD tree~$T$ of~$G$ if every~$T$-path starting in~$r$ is induced in~$G$.
Such a vertex always exists according to \Cref{lem:Dankelmann}.
We call~$M_i \ni r$ the \emph{root module} of~$G$.
We call a spanning tree~$T$ of~$G$ a \emph{poly-star} if for every module~$M_j \in P$ there is at most one vertex~$d_j$ with~$\degree_T(d_j)\geq 2$.
In \Cref{lem:internal_edges,lem:unique_high_deg_except_root,lem:unique_high_deg} we now establish that a poly-star \MAD always exists.
These lemmas are depicted in \Cref{fig:lem_no_internal_edges,fig:lem_unique_high_deg}.
\begin{figure}[t]
    \centering
    \begin{tikzpicture}[xscale=0.6, yscale=0.45]
        \begin{scope}
        \begin{scope}[on background layer]
        \draw[draw=none, fill=black!10, rounded corners] (0, 0) rectangle (2, 6);
        \end{scope}

        \node[vertex_small, color=Apricot, label=above:$r$] (r1) at (1,2) {};
        \node(Gi) at (-0.75,3) {$M_i$};

        \node[vertex_small, label=below:$x$] (x1) at (5,1) {};

        \begin{scope}[on background layer]
        \draw[draw=none, fill=black!10, rounded corners] (6, 0) rectangle (8, 6);
        \end{scope}
        \node[vertex_small, label=right:$y$] (y1) at (6.5,1) {};
        \node[vertex_small, label=right:$u$] (u1) at (7.3,3.65) {};
        \node[vertex_small, label=right:$v$] (v1) at (6.85,5.15) {};
        \node(Gj) at (8.75,3) {$M_j$};

        \draw (x1)--(y1) (u1)--(v1);
        \draw (r1) to[out=45, in=180] (3.5,5)
                            to[out=0,  in=90] (3.25,2)
                            to[out=-90, in=180] (x1);
        \draw (y1) to[out=45, in=-90] (7.3,2.3)
                            to[out=90,  in=-90] (6.5,3)
                            to[out=90, in=-135] (u1);
        \draw[thick, dashed, red] (x1)--(v1);
        \end{scope}

        \draw(10,-1)--(10,7);

        \begin{scope}[xshift=12cm, yshift=2cm]
        \begin{scope}[on background layer]
            \draw[draw=none, fill=black!10, rounded corners] (0, 0) rectangle (6, 2);
        \end{scope}
        \node[vertex_small, label=left:$u$] (u) at (1.5,1) {};
        \node[vertex_small, label=right:$v$] (v) at (4.5,1) {};
        \node(Gj) at (6.75,1) {$M_j$};

        \node[vertex_small, label=left:$u'$] (u') at (1.5,3) {};
        \node[vertex_small, label=right:$v'$] (v') at (4.5,3) {};
        \node(P) at (3,4.5) {$P^T_{u,v}$};
        \filldraw[draw=black, fill=white] (1.5,-0.5) -- (0.7,-3) -- (2.3,-3) -- cycle;
        \filldraw[draw=black, fill=white] (4.5,-0.5) -- (3.7,-3) -- (5.3,-3) -- cycle;
        \node at (1.5,-2.5) {$T_u$};
        \node at (4.5,-2.5) {$T_v$};
        \draw (u)--(u') (v)--(v') (u) -- (1.5,-0.5) (v) -- (4.5,-0.5);
        \draw (u') to[out=45, in=-180] (2.25,4)
                            to[out=0,  in=180] (3,2.75)
                            to[out=0,  in=180] (3.75,3.75)
                            to[out=0, in=180] (v');
        \draw[thick, dashed, italyGreen] (u)--(4.5,-0.5) (1.5,-0.5)--(v);
        \end{scope}
    \end{tikzpicture}
    \caption{On the left is a sketch of \Cref{lem:internal_edges}.
    If we have a \MAD~$T$ with root~$r$ (highlighted in orange), then~$T$ contains no internal edge of any other module than the root module~$M_i$.
    As sketched in the figure, if there were such an edge~$uv$ in~$T$, then the path from~$r$ to~$v$ would not be induced in~$G$ due to the dashed red edge which contradicts the assumption that~$r$ is a root.\\
    On the right is a sketch of \Cref{lem:unique_high_deg_except_root}.
    If we have a \MAD~$T$ with two non-leaves~$u$ and~$v$ in a module~$M_j$, which is not the root module, then repositioning one of the subtrees~$T_u$ or~$T_v$ by one of the dashed green edges yields a better spanning tree.
    Note that~$u$ and~$v$ have the same neighborhood outside their module and by \Cref{lem:internal_edges} they have no neighbors inside their module.}
    \label{fig:lem_no_internal_edges}
\end{figure}
\begin{lemma}
    \label{lem:internal_edges}
    Let~$T$ be a MAD tree of~$G$ with root~$r \in M_i$.
    There is no edge~$uv \in E(T)$ such that~$u,v \in M_j$ for any~$M_j \in P \setminus \{M_i\}$.
\end{lemma}
\begin{proof}
    Assume towards a contradiction that~$T$ contains an edge~$uv$ inside a module~$M_j$ with~$j\neq i$.
    W.l.o.g.~assume $\dist_T(r,u)<\dist_T(r,v)$ and consider the path~$P^T_{r,v}$.
    This path contains vertices~$x$ and~$y$ such that~$x \notin M_j, y \in M_j$ and~$xy \in E(T)$.
    By modularity, it holds that also~$xv \in E(G)$.
    Hence,~$P^T_{r,v}$ is not induced in~$G$ which contradicts the assumption that~$r$ is a root of~$T$.
\end{proof}

\begin{lemma}
    \label{lem:unique_high_deg_except_root}
    Let~$T$ be a MAD tree of~$G$ with root~$r \in M_i$.
    Then in any module~$M_j \in P \setminus \{M_i\}$ there is at most one vertex~$d \in M_j$ with~$\degree_T(d) \geq 2$.
\end{lemma}
\begin{proof}
    Assume towards a contradiction that there are two vertices~$u,v \in M_j$ with~$\degree_T(u) \geq 2$ and~$\degree_T(v) \geq 2$.
    Since~$T$ is connected, there is the path~$P^T_{u,v}$ in~$T$.
    Let~$u',v' \in V(P^T_{u,v})$ such that~$u \in N_T(u')$ and~$v \in N_T(v')$ and
    let~$V_u \coloneqq N_T(u)\setminus \{u'\}$ and~$V_v \coloneqq N_T(v)\setminus\{v'\}$.
    By \Cref{lem:internal_edges}, we know that~$V_u \cap M_j = V_v \cap M_j = \emptyset$ and therefore, by modularity, $V_u,V_v \subseteq  N_G(v)$ and~$V_u, V_v \subseteq N_G(u)$.
    Using \Cref{obs:common_anchor}, we immediately get a contradiction, since there exists a spanning tree~$T'$ with~$\W(T')<\W(T)$.
\end{proof}

\begin{figure}[t]
    \centering
    \begin{tikzpicture}[xscale=0.65, yscale=0.5]
        \begin{scope}[on background layer]
        \draw[draw=none, fill=black!10, rounded corners] (0, 0) rectangle (2, 6);
        \end{scope}
        \node[vertex_small, color=Apricot, label=left:$r$] (r) at (1,1) {};
        \node[vertex_small, label=right:$w$] (w) at (2.75,3.25) {};
        \node[vertex_small, label=above:$v$] (v) at (1,4.25) {};
        \node[vertex_small, label=left:$v'$] (v') at (-1.25,3.75) {};
        \node[vertex_small, label=left:$v''$] (v'') at (1,2.75) {};

        \draw (v)--(w) (v)--(v') (v)--(v'');
        \draw[thick, dashed, red] (r)--(v') (w)--(v'');
        \draw (r) to[out=-30, in=210] (3,0.5)
                            to[out=30,  in=-120] (2.25,2)
                            to[out=30, in=-90] (w);

        \node(case1) at (1,-1) {Case 1};
        \node(Gi1) at (1,6.5) {$M_i$};

        \begin{scope}[on background layer]
        \draw[draw=none, fill=black!10, rounded corners] (8, 0) rectangle (10, 6);
        \end{scope}
        \node[vertex_small, color=Apricot, label=right:$r$] (r2) at (9,1) {};
        \node[vertex_small, label=left:$x$] (x) at (7.25,2.3) {};
        \node[vertex_small, label=right:$v$] (v2) at (9,2.5) {};

        \node[vertex_small, label=above:$y$] (y) at (9,3.8) {};
        \filldraw[draw=black, fill=white] (9, 3.8) -- (8.4,5.2) -- (9.6,5.2) -- cycle;

        \draw (r2)--(x) (r2)--(v2) (v2)--(9,3.8);
        
        \draw[thick, dashed, italyGreen] (x)--(9.2,4.3);
        \draw[thick, dashed, italyGreen] (x)--(9,4.6);
        \draw[thick, dashed, italyGreen] (x)--(8.8,4.9);
        \draw[thick, dashed, italyGreen] (x)--(8.6,5.2);

        \node(case2) at (9,-1) {Case 2};
        \node(Gi1) at (9,6.5) {$M_i$};
    \end{tikzpicture}
    \caption{A sketch of \Cref{lem:unique_high_deg}.
    Suppose we have a \MAD~$T$ with root~$r$ (highlighted in orange) and another vertex~$v$ from the root module~$M_i$ which also has degree at least two in~$T$.
    Then there are two cases.
    \textbf{Case~1:} If the path from~$r$ to~$v$ in~$T$ leaves the root module~$M_i$, then the paths from~$r$ to~$v'$ and to~$v''$ are not induced as indicated by the dashed red edges.
    \textbf{Case 2:} Otherwise, we can assume that~$v$ is a neighbor of~$r$ in~$T$. In this case we obtain another \MAD by appending the neighbors of~$v$ in a star-like manner to a neighbor~$x$ of~$r$ outside of~$M_i$.
    This is indicated by the dashed green edges.}
    \label{fig:lem_unique_high_deg}
\end{figure}

\begin{lemma}
    \label{lem:unique_high_deg}
    There exists a MAD tree~$T^*$ of~$G$ which is a poly-star, meaning in every module of~$P$ there is at most one vertex with degree larger than one in~$T^*$.
\end{lemma}
\begin{proof}
    By \Cref{lem:unique_high_deg_except_root}, we know that there is a MAD tree~$T$ of~$G$ with root~$r \in M_i$ such that the claim holds for any module~$M_j \in P\setminus\{M_i\}$.

    We now consider the root module~$M_i$ and show the existence of a \MAD~$T^*$ with the desired properties.
    Let~$v \in M_i\setminus\{r\}$ be a vertex with~$\degree_T(v) \geq 2$.
    We consider the two cases~$V(P^T_{r,v})\cap M_i\neq V(P^T_{r,v})$ and~$V(P^T_{r,v}) \cap M_i = V(P^T_{r,v})$.

    For the first case, we show that no such~$v$ can exist.
    Assume the contrary.
    Then there must be some~$w \in V(P^T_{r,v})\setminus M_i$ and~$u \in V(P^T_{r,v})\cap M_i$ with~$wu\in E(P^T_{r,v})$.
    Consider some~$v' \in N_T(v)\setminus V(P^T_{r,v})$, which must exist since by assumption~$\degree_T(v)\geq 2$.
    If~$v' \in M_i$, then the path~$P^T_{r,v'}$ is not induced in~$G$ because of the edge~$wv'$, which exists by modularity.
    If~$v' \notin M_i$, then the path~$P^T_{r,v'}$ is not induced in~$G$ because of the edge~$rv'$, which exists by modularity.
    Either case produces a contradiction to the fact that~$r$ is root of~$T$.
    
    For the second case, we show how~$T^*$ can be obtained from~$T$.
    We assume that~$v\in N_T(r)$ as otherwise we can consider the unique vertex from~$N_T(r) \cap P^T_{r,v}$ instead.
    Let~$Y \coloneqq N_T(v)\setminus\{r\}$ and~$T_Y \coloneqq \bigcup_{u\in Y}T^u_{vu}$.
    If~$T_Y$ contains a vertex~$w$ which is not contained in the root module~$M_i$, then~$P^T_{r,w}$ is not induced in~$G$.
    This implies~$V(T_Y) \subseteq M_i$.
    Since~$k \geq 2$, there exists~$x \in N_T(r)\setminus M_i$, because~$T$ is connected and~$P^T_{r,x}$ is induced in~$G$.
    In particular, if the module of~$x$ contains an internal vertex of~$T$, then let~$x$ be this internal vertex (it still holds that~$x \in N_T(r)\setminus M_i$, as otherwise~$P^T_{r,x}$ would be non-induced in~$G$).
    Using modularity, we obtain~$T^*$ from~$T$ by repositioning all vertices from~$T_Y$ as leaves to~$x$:
    \[E(T^*) \coloneqq \Big(E(T)\setminus \big(E(T_Y) \cup \{uv\mid u \in Y\}\big)\Big)\cup \{ux\mid u \in V(T_Y)\}.\]
    Let~$A := V(T_Y) \cup \{x,v\}$ and~$B:=V(G)\setminus A$.
    Observing that a star attains minimum Wiener index; that in~$T$ the distance to~$x$ is no larger than the distance to~$v$ for any vertex in~$B$; and that distances within~$B$ remain unaffected by the changes, we get:
    \begin{align*}
    \dist_{T^*}(A,A) \leq & \dist_{T}(A,A),\\
    \dist_{T^*}(A,B) \leq & \dist_{T}(A,B),\\
    \dist_{T^*}(B,B) = & \dist_{T}(B,B),
    \end{align*}

    and conclude~$\W(T^*)\leq\W(T)$. Hence~$T^*$ is a MAD tree.
    By iteratively applying this argument to all neighbors of~$r$ from~$M_i$ in~$T$ with degree at least two, we obtain the desired \MAD.
\end{proof}

We now further restrict the structure of a poly-star \MAD, by specifying adjacency of the leaves.
We show that all vertices of a non-root module are adjacent to a single vertex in the adjacent module closest to the root module.
All vertices in the root module that are not adjacent to the root are connected to a single neighbor in an adjacent module -- the neighbor that minimizes the distance to ``the rest'' of the graph.
Note that the proof of \Cref{lem:unique_high_deg} implies that we can always assume that in the root module the only vertex of degree at least two is the root.

\begin{restatable}{lemma}{commonneighbor}
    \label{lem:common_neighbor}
    Let~$T$ be a poly-star MAD tree of~$G$ with root~$r\in M_i$.
    Let~$d_j \in M_j$ be such that~$\degree_T(d_j) \geq \degree_T(v)$ for all $v \in M_j$ and let~$D\coloneqq\{d_1,\ldots,d_i=r,\dots,d_k\}$.
    Then the following holds:
    \begin{enumerate}[i)]
    \item\label{cn_tree} $T[D]$ is a tree.
    \item\label{cn_non_root} For every non-root module~$M_j \in P \setminus \{M_i\}$, it holds for all leaves~$v \in M_j\setminus\{d_j\}$ that~$N_T(v) = \{c_j\}$, where~$c_j$ is the unique vertex in~$P^{T[D]}_{r,d_j}\cap N_T(d_j)$.
    \item\label{cn_root} Let~$R = T - (M_i\setminus N_T[r])$.
    For all leaves~$v \in M_i\setminus N_T[r]$ from the root module which are not adjacent to the root, it holds that~$N_T(v) = \{c_i\}$, where~$c_i\in\argmin\limits_{w \in N_T(r)\cap D}\dist_{R}(w)$.
    \end{enumerate}
\end{restatable}
\begin{proof}
    We show the statements separately.
    \begin{enumerate}[(i),wide,itemsep=5pt]
     \item For any~$v \notin D$, it holds that~$\degree_T(v) = 1$, because~$T$ is a poly-star.
    Clearly, since~$T$ is a tree, also~$T[D]$ is a tree and the unique neighbor of any~$v \notin D$ is in~$D$.

    \item Consider module~$M_j \in P\setminus\{M_i\}$.
    Let~$C = \bigcup_{u \in M_j\setminus\{d_j\}}N_T(u)$.
    By \Cref{lem:internal_edges}~$C \cap M_j = \emptyset$.
    We show that all leaves $u \in M_j\setminus\{d_j\}$ are adjacent to the same vertex. Assume for a contradiction~$|C|>1$ and let~$c \in \argmin_{u\in C} \dist_{T}(u)$.
    Applying \Cref{obs:common_anchor} repeatedly we obtain~$T'$ with~$N_{T'}(v) = \{c\}$ for any~$v \in M_j\setminus \{{d_j\}}$ and~$\W(T')<\W(T)$, contradicting the optimality of $T$.
    We now show~$c \in N_T(d_j)$, so assume the opposite for contradiction and choose~$c' \in N_T(d_j) \setminus M_j$ such that~$c' \in V(P^T_{d_j,c})$.
    Using \Cref{obs:common_anchor} we immediately get that there is a spanning tree of~$G$ of smaller Wiener index (to apply \Cref{obs:common_anchor} choose~$a=c'$,~$V_a = \{d_j\}$,~$b=c$ and~$V_b = M_j\setminus\{d_j\}$), contradicting the optimality of $T$.

    Finally we show that~$c$ is the unique vertex in~$P^T_{r,d_j}\cap N_T(d_j)$.
    Assume towards a contradiction~$c \notin P^T_{r,d_j}\cap N_T(d_j)$.
    Denote with~$c'$ the unique vertex in~$P^T_{r,d_j}\cap N_T(d_j)$ and let~$v \in M_j\setminus\{d_j\}$.
    The path~$P^T_{r,v}$ is not induced in~$G$, because of the edge~$c'v\in E(G)$, which exists by modularity, contradicting the assumption that~$r$ is root of~$T$.

    \item Consider module~$M_i$.
    Let~$C = \bigcup_{u \in M_i\setminus N_T[r]}N_T(u)$.
    It holds that~$C \cap M_i = \emptyset$, since~$r\notin C$ by definition and an edge from any~$u\in M_i\setminus N_T[r]$ to any~$v \in M_i\setminus\{r\}$ would imply~$\degree_T(v)\geq 2$ for~$T$ to be connected, contradicting the fact that~$T$ is a poly-star.
    By the same arguments as in case~(ii) we get that there is a vertex~$c_i \in N_T(r)\setminus M_i$ such that~$N_T(v) = \{c_i\}$ for all~$v \in M_i\setminus N_T[r]$.
    Finally we show 
    \[c_i\in\argmin\limits_{w \in N_T(r)\cap D}\dist_{R}(w).\]

    It must hold that~$c_i \in N_T(r)\cap D$, since~$c_i \in N_T(r)$ and all leaves in~$V(G)\setminus D$ are connected to vertices in~$D$.
    Assuming there was a vertex~$c_i' \in N_T(r)\cap D$ with~$\dist_{R}(c_i')<\dist_{R}(c_i)$, clearly a spanning tree of strictly smaller Wiener index than~$T$ could be constructed by replacing any edge~$vc_i\in E(T)$ with~$v \in M_i\setminus\{r\}$ by an edge~$vc_i'$ (this follows from the proof of \Cref{obs:common_anchor}).
    \end{enumerate}
\end{proof}
To determine the exact layout of the poly-star \MAD it remains to consider which vertices to connect to the root and how to choose the root vertex within the root module.
For the first part, note that for a \MAD~$T$ of~$G$ with root~$r \in M_i$ it holds that~$N_T(r) = N_G(r)$, as any~$v \in N_G(r)\setminus N_T(r)$ would create a non-induced $T$-path starting in~$r$.
For the second part, we argue that the root~$r$ has maximum degree in~$G$ among the vertices of~$M_i$.

\begin{restatable}{lemma}{rootmaxdegree}
    \label{lem:root_max_degree}
    There exists a poly-star \MAD~$T^*$ of~$G$ with root~$r^*\in M_{i^*}$, such that~$\degree_G(r^*) \geq \degree_G(v)$ for all~$v \in M_{i^*}$.
\end{restatable}
\begin{proof}
    Let~$T$ be a poly-star \MAD of~$G$ with root~$r \in M_i$.
    Let~$d_j \in M_j$ be such that~$\degree_T(d_j) \geq \degree_T(v)$ for all $v \in M_j$ and let~$D\coloneqq\{d_1,\ldots,d_i=r,\dots,d_k\}$.
    We first establish that w.l.o.g. we can assume~$\degree_G(d_j)\geq\degree_G(v)$ for all~$v \in M_j$ for any module~$M_j \in P$.
    First consider any non-root module~$M_j \in P\setminus\{M_i\}$.
    Let~$v$ be a maximum-degree vertex in~$M_j$ with respect to~$G$.
    In case~$d_j = v$ we are done.
    Otherwise~$T^*$ isomorphic to~$T$ can be obtained by exchanging the neighborhood of~$d_j$ and~$v$, since according to \Cref{lem:internal_edges}~$N_T(d_j)\cap M_j = N_T(v) \cap M_j = \emptyset$ and by modularity~$N_G(d_j)\setminus M_j = N_G(v)\setminus M_j$.

    Now consider root module~$M_i$.
    In case~$r = d_i$ has maximum degree within~$M_i$ with respect to~$G$, we are done.
    Otherwise assume there is a vertex~$v \in M_i$ with~$\degree_G(v) > \degree_G(r)$.
    Let~$c \in N_T(r) \setminus M_i$ such that~$N_T(u) = \{c\}\ \forall u \in M_i \setminus N_T[r]$.
    Such~$c$ exists according to \Cref{lem:common_neighbor}.
    Let~$x = |N_T(r)\cap M_i|$.
    Consider~$T^*$ in which~$N_{T^*}(v)\setminus M_i =  N_{T}(r) \setminus M_i$ (which is possible by modularity) and in which exactly~$x$ vertices in~$M_i$ are adjacent to~$v$, while all other~$|M_i| - x - 1$ vertices are adjacent to~$c$.
    This is possible since~$|N_G(v) \cap M_i| > x$ as~$|N_G(v)|>|N_G(r)|$ by assumption and~$N_G(v) \setminus M_i = N_G(r) \setminus M_i$ by modularity.
    $T^*$ is isomorphic to~$T$.

    Since~$T^*$ is isomorphic to~$T$, it is a poly-star \MAD of~$G$.
    In~$T^*$ any internal node~$v \in M_j$ has maximum degree with respect to~$G$ within its module~$M_j$. In any \MAD there is always a root of degree at least two, because in case a leaf is root then so is its unique neighbor.
    Since in~$T^*$ any vertex of degree at least two is a maximum-degree vertex within its module, there exists a root~$r^* \in M_{i^*}$ of~$T^*$ such that~$\degree_G(r^*)\geq \degree_G(v)$ for any~$v \in M_{i^*}$.
\end{proof}

Combining the structural observations above finally proves \Cref{thrm:mw_fpt}.
\begin{proof}[Proof of \Cref{thrm:mw_fpt}]
    Let~$T^*$ be a poly-star MAD tree of~$G$ with root~$r^* \in M_{i^*}$, such that~$N_{T^*}(r^*) = N_G(r^*)$ and for any~$v \in M_{i^*}~\degree_{G}(r^*)\geq\degree_{G}(v)$.
    This exists by \Cref{lem:root_max_degree} and the observation made just before it.
    Let~$d_j^* \in M_j$ be such that~$\degree_{T^*}(d_j^*)\geq \degree_{T^*}(v)$ for all $v \in M_j$ and let~$D^*\coloneqq\{d_1^*,\ldots,d_i^*=r^*,\dots,d_k^*\}$.

    Consider the iteration of~\Cref{alg:mw} where $i=i^*$ and~$T'$ on~$G_{/P}$ is isomorphic to~$T^*[D^*]$, which is a tree according to \Cref{lem:common_neighbor}.
    \Cref{alg:mw} creates vertex set~$D$ and initializes~$T[D]\cong T^*[D^*]$ and chooses~$r$ such that~$\degree_G(r)\geq \degree_G(v)$ for all~$v \in M_i$.
    The algorithm includes the entire $G$-neighborhood of~$r$ in~$T$, ensured by Lines~\ref{alg_mw:t_loop} and~\ref{alg_mw:root_neighbors}.
    Any remaining leaves are connected according to  \Cref{lem:common_neighbor} as ensured in Lines~\ref{alg_mw:cn_a1} to~\ref{alg_mw:cn_a2} and~\ref{alg_mw:cn_b1} to~\ref{alg_mw:cn_b2}.
    If~$r\neq r^*$, then \Cref{alg:mw} creates a \MAD of~$G$ that is isomorphic to~$T^*$.
    Similarly, if~$d_{j}^*\neq d_j$ for any~$j \in [k]\setminus \{i\}$, then \Cref{alg:mw} creates a \MAD isomorphic to~$T^*$ since~$T[D]$ is isomorphic to~$T^*[D^*]$ and \Cref{alg:mw} connects all~$v \in M_j\setminus \{d_j\}$ according to \Cref{lem:common_neighbor}.

    Since \Cref{alg:mw} tries all possible spanning trees~$T'$ combined with all possible choices of a root module~$M_i$, it must at some point choose~$i=i^*$ and~$T'$ isomorphic to~$T^*[D^*]$ and thus finds a \MAD of~$G$.

    Let~$n=n_G$ and~$m=m_G$.
    A modular decomposition and thus also a modular partition can be computed in~$O(n+m)$ time~\cite{TCHP08}.
    The number of choices for the root module~$i$ is~$k$.
    There are at most~$k^{k-2}$ possible spanning trees on the graph~$G_{/P}$~\cite{C89}, which can be enumerated in~$O(k^{k-1})$ time~\cite{GM78,W86}.
    Using breadth-first search (BFS), the paths from~$q_i$ to all other vertices in the tree~$T'$ can be found in~$O(k)$ time and the neighbor closest to~$q_i$ can be stored for any vertex.
    A maximum-degree vertex for every module of~$G$ can be found in~$O(n+m)$ time and the Wiener index of a tree can also be computed in~$O(m+n)$ time~\cite{MP88}.
    The distances of a single vertex to all others can also be computed in~$O(n+m)$ time using BFS.
    This gives an overall running time of~$O(k^{k+1}(m+n))$.
\end{proof}

The minimum number~$k$ of modules in a modular partition of a graph~$G$ is bounded by the modular width of~$G$, but can in general be arbitrarily smaller (consider two modules both consisting of a long path).
Parameter~$k$ lends itself particularly well to the \MADST problem as modularity immediately gives a lower bound on the distance between vertices and hardly any information on connectivity within a module is required.
While this seems highly problem-specific, it could be of interest to find algorithms for other problems utilizing only the structure of a modular partition or to exclude such algorithms (by showing W-hardness) for problems that are FPT with respect to the modular width.

\section{Parameters Larger Than Treewidth}
\label{sec:larger-params}
In \Cref{sec:width-params}, the question whether \MADST is fixed-parameter tractable with respect to treewidth is left open.
Exploring the border of tractability for \MADST, we looked for larger parameters, that is, parameters that can be lower bounded by a function of the treewidth, but that can in general not be upper bounded by any function depending solely on the treewidth.
We obtain fixed-parameter tractability for three such parameterizations: vertex integrity, max leaf number and an above-guarantee parameter; which we discuss now.

\subsection{Vertex Integrity}
    A graph has small vertex integrity if it can be broken into small components by removing a small number of vertices.
    Formally, the vertex integrity~$k$ of a graph~$G$ is~$k \coloneqq \min_{S \subseteq V(G)} \{|S| + \max_{C \in \cc(G-S)} |C| \}$ where~$\cc(G-S)$ denotes the set of all connected components in~$G-S$.
    We remark that computing a witness~$S$ for~$k$ is fixed-parameter tractable with respect to~$k$~\cite{DH16}.

    Algorithms exploiting vertex integrity are often based on the observation that the number of different (i.e.\ non-isomorphic) components in~$G-S$ is bounded in a function of~$k$. 
    This is often combined with a problem-specific observation stating that a ``solution'' can be assumed to be the same in two isomorphic components in~$G-S$.
    While such an observation seems plausible for \MADST as well, proving it turns out to be non-trivial.
    To this end, assume that the edges of a \MAD~$T$ within~$G[S]$ are given (this can be brute-forced in~$f(k)$ time).
    Further assume for simplicity that~$T[S]$ is a tree and that two connected components~$C_1$ and~$C_2$ in~$G-S$ are isomorphic and have exactly the same neighbors in~$S$.
    A natural approach is to show that if the solution~$T$ is different within~$C_1$ and~$C_2$, then we simply change the solution within one component so that they are identical.
    The crux here is that this changes the distances between the vertices from~$C_1$ and~$C_2$, and one needs to argue that the sum of these distances will not increase; this seems non-trivial.
    
    We encountered similar situations when designing our algorithm for modular width.
    The difference is that for two adjacent modules \emph{each} vertex in one module is adjacent to \emph{all} vertices in the other module.
    We do not have such strong structural properties here.
    Thus, we employ an integer quadratic program to compute the best way to extend the tree on~$S$ to the connected components in~$G-S$, obtaining the following.
    \begin{restatable}{theorem}{fptvertexintegrity}
		\label{thm:fpt-vertex-integrity}
        \MADST is in FPT with respect to the vertex integrity.
    \end{restatable}

    \begin{proof}
        Let~$(G,b)$ be a given \MADST-instance and~$S \subseteq V(G)$ be a witness for the vertex integrity of the graph~$G$, that is, $k = |S| + \max_{C \in \cc(G-S)} |C|$.
        The basic idea of the algorithm is to iterate over all possible forests on~$S$, extend them to trees and then add all remaining components to these trees in an optimal way via integer quadratic programming.

        We say two components~$C_1$ and~$C_2$ in~$G-S$ are of the same \emph{type} if they are isomorphic with isomorphism~$\phi\colon V(C_1) \rightarrow V(C_2)$ such that~$N_G(u) \cap S = N_G(\phi(u)) \cap S$ for all~$u \in V(C_1)$.
        Since the size of~$S$ and the size of each component is bounded by~$k$, there are at most~$2^{2k^2}$ types of components.

        Let~$F_S \subseteq G$ be a forest on~$S$ with components~$F_1, \dots, F_r$.
        If~$F_S$ is a subgraph of a spanning tree~$T$ of~$G$, then there exist up to~$r-1$ components~$C_1, \dots, C_{r-1}$ in~$G-S$ and a tree~$T_S$ such that~$T_S =T[S \cup V(C_1) \cup \dots \cup V(C_{r-1})]$.
        For each~$i \in \{1,\dots, r-1\}$, there are at most~$2^{2k^2}$ possible types for the component~$C_i$, $2^{k^2}$ possible forests on~$V(C_i)$, and~$2^{k^2}$ possible edge sets between~$S$ and~$V(C_i)$.
        Since~$r \leq k$, our algorithm only has to consider at most~$2^{4k^3}$ possible trees~${T_S \supseteq F_S}$ for every forest~$F_S$ on~$S$.
        Every such tree~$T_S$ contains at most~$k-1$ added components of size at most~$k$, hence~$|V(T_S)| \leq k^2$.
        
        In the last step, the algorithm further extends~$T_S$ in the best possible way to a spanning tree of~$G$.
        This is done with integer quadratic programming.

        Suppose~$T \supseteq T_S$ is a spanning tree of~$G$.
        Two components~$C_1,C_2$ in~$G-V(T_S)$ (hence also in~$G-S$) of the same type~$t$ induce the same \emph{extension}~$\mathcal E$ of~$T_S$ in~$T$ if there exists an isomorphism~$\phi$ from~$T[V(T_S) \cup V(C_1)]$ to~$T[V(T_S) \cup V(C_2)]$ such that~$\phi (u) = u$ for all~$u \in V(T_S)$.
        Since~$S$ has size at most~$k$ and each component in~$G-S$ has size at most~$k$, there are at most~$2^{2k^2}$ extensions~$\mathcal E$ of~$T_S$ in~$T$ for each component type~$t$.
        
        For every component type~$t$ and every possible extension~$\mathcal E$ of~$T_S$ in a spanning tree by a component of type~$t$, we introduce an integer variable~$x_{t,\mathcal E}$.
        Intuitively, the value of~$x_{t,\mathcal E}$ is the number of components of type~$t$ which induce the same extension~$\mathcal E$ of~$T_S$ in a MAD tree~$T \supseteq T_S$.
        Moreover, we need the following constants (which clearly can be computed in FPT time with respect to~$k$ since the size of all components and~$T_S$ is bounded by~$k^2$):
        Let~$C$ and~$C'$ be components of types~$t$ and~$t'$ respectively in~$G-V(T_S)$ and let~$T_S^{\mathcal E,\mathcal E'}$ be the tree which is obtained by extending~$T_S$ with~$C$ and~$C'$ according to extensions~$\mathcal E$ and~$\mathcal E'$ respectively.
        We define
        \begin{compactitem}
            \item $D_{t,\mathcal E} \coloneqq \dist_{T_S^{\mathcal E, \mathcal E'}}(V(C), V(C))$,
            \item $D_{t,\mathcal E,t',\mathcal E'} \coloneqq \dist_{T_S^{\mathcal E, \mathcal E'}}(V(C), V(C'))$,
            \item $D_{T_S,t,\mathcal E} \coloneqq \dist_{T_S^{\mathcal E, \mathcal E'}}(V(C), V(T_S))$, and
            \item $n_t$ as the number of components with type~$t$ in~$G-V(T_S)$.
        \end{compactitem}
        The algorithm precomputes these values for all types and their possible extensions of~$T_S$.
        Finding a spanning tree~$T \supseteq T_S$ which minimizes the Wiener index can then be formulated as the following integer quadratic program:
        \begin{align}
            \min \quad & \sum_{t,\mathcal E} \sum_{t',\mathcal E'} x_{t,\mathcal E} x_{t',\mathcal E'} D_{t,\mathcal E,t',\mathcal E'} + \sum_{t,\mathcal E} x_{t,\mathcal E} (D_{t,\mathcal E} + D_{T_S,t,\mathcal E}) \notag \\
            \text{s.t.} \quad 
            & \sum_{\mathcal E} x_{t,\mathcal E} = n_t \quad && \forall t \label{cons:component_numbers} \\
            & 0 \leq x_{t,\mathcal E} \leq n_t \quad && \forall t, \mathcal E. \label{cons:max_values}
        \end{align}
        If we add the value~$\W(T_S)$ to the objective function of the IQP, then we obtain the Wiener index of the spanning tree~$T$ which is obtained by extending~$T_S$ according to the variable values.
        This means $x_{t,\mathcal E}$ components of type~$t$ in~$G - V(T_S)$ extend~$T_S$ in~$T$ according to~$\mathcal E$.
        The first constraint ensures that we use all components of every type and the second constraint ensures that the values of the variables are realizable in our graph.

        Finally, the algorithm returns ``yes'' if for one of the considered trees~$T_S$ the above integer quadratic program has an optimal value of at most~$b - \W(T_S)$.

        Since the coefficients of the objective function as well as the number of variables are bounded by some function of~$k$, the above program is solvable in FPT time with respect to~$k$~\cite{Lokshtanov15}.
        This concludes the proof.
    \end{proof}

\subsection{Max Leaf Number}

Another rather large parameter is the max leaf number~$\ml(G)$ of a connected graph~$G$.
Formally it is defined as the maximum number of leaves in any spanning tree of~$G$.
For us it is useful because the number of vertices with degree at least three in~$G$ is~$O(\ml(G))$~\cite{BankevichKarpov2012, HanakaOOV26}.
This fact is used to prove the following theorem.

\begin{theorem}
	\label{thrm:ml_fpt}
	\MADST is in FPT with respect to the max leaf number.
\end{theorem}

\begin{proof}
Let~$(G,b)$ be a given \MADST-instance.
A \emph{trail} in~$G$ is a walk in which no edges are repeated.
Similar to \citet{HanakaOOV26} we call a trail in which all internal vertices have degree exactly two and both endpoints are contained in~$V_{\geq 3} \coloneq \{v \in V(G) \mid \deg_G(v) \geq 3\}$ a \emph{maximal degree-2 trail}.
Let~$\mathcal{T}_2$ be the set of all maximal degree-2 trails in~$G$ and note that each element in~$\mathcal{T}_2$ is either a path or an induced cycle (if the endpoints are identical) in~$G$.
Moreover, it is known that $|\mathcal T_2|\in O(\ml(G)^2)$ and that~$\mathcal{T}_2$ can be computed in polynomial time \cite{HanakaOOV26}.

Now, suppose~$T$ is a spanning tree of~$G$.
Note that each edge in~$E(G)\setminus E(T)$ is contained in some trail in~$\mathcal T_2$.
Moreover, since~$T$ is a spanning tree, it is clear that~$E(G)\setminus E(T)$ contains at most one edge from each trail in~$\mathcal{T}_2$ (otherwise, $T$ would not be connected).
In the first step, our algorithm guesses from which elements in~$\mathcal{T}_2$ edges are removed to obtain a \MAD.
We call this set of trails~$\mathcal{S} \subseteq \mathcal{T}_2$ and remark that there are~$2^{O(\ml(G)^2)}$ possible such sets.
By~$\internal(\mathcal{S})$ we denote the set of all internal vertices in the trails from~$\mathcal{S}$ and by~$E(\mathcal{S})$ the set of all edges in~$\mathcal{S}$.
If~$G_\mathcal S\coloneqq G - E(\mathcal{S}) -\internal(\mathcal{S})$ is not a tree, then we cannot obtain a spanning tree of~$G$ by deleting one edge in each trail from~$\mathcal{S}$.
Consequently, we can discard the current set~$\mathcal{S}$ in this case and continue with the next subset of~$\mathcal{T}_2$.
See \Cref{fig:trails} for an illustration.

\begin{figure}[t]
    \centering
    \begin{tikzpicture}[xscale=0.9, yscale=0.6]
    
        \node[vertex_small, color=CornflowerBlue] (a) at (0,2) {};
        \node[vertex_small, color=Apricot, inner sep=2pt] (b) at (0,5) {};
        \node[vertex_small, color=Apricot, inner sep=2pt] (c) at (4,6) {};
        \node[vertex_small, color=Apricot, inner sep=2pt] (d) at (5,2) {};
        \node[vertex_small, color=Apricot, inner sep=2pt] (e) at (3,0) {};
        \node[vertex_small, color=CornflowerBlue] (v1) at (-1,5) {};
        \node[vertex_small, color=CornflowerBlue] (v2) at (-1,6) {};
        \node[vertex_small, color=CornflowerBlue] (v3) at (0,6) {};
        \node[vertex_small] (v4) at (1.33,5.33) {};
        \node[vertex_small] (v5) at (2.66,5.66) {};
        \node[vertex_small] (v7) at (4.5,1.5) {};
        \node[vertex_small] (v8) at (4,1) {};
        \node[vertex_small] (v9) at (3.5,0.5) {};
        \node[vertex_small, color=CornflowerBlue] (v10) at (1.5,0.25) {};
        \node[vertex_small, color=CornflowerBlue] (v11) at (0,0.5) {};
        \node[vertex_small, color=CornflowerBlue] (v12) at (0,3.5) {};
        \node[vertex_small] (v13) at (1,4.4) {};
        \node[vertex_small] (v14) at (2,3.8) {};
        \node[vertex_small] (v15) at (3,3.2) {};
        \node[vertex_small] (v16) at (4,2.6) {};
        \node[vertex_small, color=CornflowerBlue] (v18) at (5.66,4.13) {};
        \node[vertex_small, color=CornflowerBlue] (v19) at (5.33,5.46) {};
        \node[vertex_small, color=CornflowerBlue] (v20) at (1.5,2.5) {};
        
        \node[text=CornflowerBlue] () at (-1,2.5) {\Large$\mathcal{S}$} ;

        \draw[thick] (b)--(v4)--(v5)--(c) (d)--(v7)--(v8)--(v9)--(e) (b)--(v13)--(v14)--(v15)--(v16)--(d);
        \draw[thick, CornflowerBlue] (c)--(d) (c)--(v19)--(v18)--(d) (b)--(v1)--(v2)--(v3)--(b) (e)--(v10)--(v11)--(a)--(v12)--(b) (e)--(v20)--(b);
    \end{tikzpicture}
    \caption{The figure illustrates various notations from the proof of \cref{thrm:ml_fpt} with an example graph.
    The vertices with degree at least three are orange.
    The set~$\mathcal{T}_2$ contains eight elements (three black and five blue trails).
    Here, the blue trails form one possible set~$\mathcal{S}$ such that the removal of the corresponding trails yields a tree.
    After guessing~$\mathcal{S}$, the algorithm finds an optimal way of removing one edge from each trail in~$\mathcal{S}$ via integer quadratic programming.}
    \label{fig:trails}
\end{figure}

Now, consider a fixed set~$\mathcal{S} = \{ P_1, \dots, P_k\}$ of trails such that~$G_\mathcal S$ is indeed a tree and let~$n_i$ be the number of internal vertices in~$P_i$.
The idea is to formulate an integer quadratic program (IQP) that minimizes the Wiener index among all spanning trees which remove exactly one edge from each trail in~$\mathcal{S}$.
Since~$G_\mathcal{S}$ is a tree, we can compute the pairwise distances (in $G_\mathcal S$) between all its vertices~$R_\mathcal{S} \coloneq V(G) \setminus \internal(\mathcal{S})$ and then use them in our IQP formulation.
For~$u,v \in R_\mathcal{S}$, we denote their distance in the tree~$G_\mathcal{S}$ by~$d_{\mathcal{S}}(u,v)$ and the total sum of pairwise distances in $R_\mathcal S$ by~$D_{R_\mathcal{S}}$.
The IQP will compute the sum of all remaining distances of vertex pairs in~$\binom{V(G)}{2}\setminus\binom{R_\mathcal S}{2}$.
Thus, our algorithm returns yes if the minimum value of the IQP is at most~$b - D_{R_\mathcal{S}}$.%

The IQP contains one variable~$x_i$ for each trail~$P_i$ in~$\mathcal{S}$.
Let~$a_i, b_i \in V_{\geq 3}$ be the two endpoints of~$P_i$.
Note that~$a_i = b_i$ if and only if~$P_i$ is a cycle in~$G$ with exactly one vertex from~$V_{\geq 3}$.
The values of the variables correspond to a spanning tree of~$G$ by removing the~$(x_i + 1)$-th edge of~$P_i$ when counting from the endpoint~$a_i$.
Consequently, in the resulting spanning tree, paths with~$x_i$ and~$(n_i - x_i)$ vertices are attached to~$a_i$ and~$b_i$ respectively.
See \Cref{fig:ml_iqp} for an illustration.
The objective of our IQP computes the pairwise distances inside~$\internal(\mathcal{S})$ and from~$\internal(\mathcal{S})$ to~$R_\mathcal{S}$.
In the following, we split the objective into three parts.

\begin{figure}[t]
    \centering
    \begin{tikzpicture}[xscale=0.9, yscale=0.6]
    
        \node[vertex_small, color=Apricot, inner sep=2pt, label=left:$a_i$] (a) at (0,2) {};
        \node[vertex_small, color=Apricot, inner sep=2pt, label=right:$b_i$] (b) at (10,2) {};
        \node[vertex_small, label=above:{$+1$}] (v1) at (1,2) {};
        \node[vertex_small, label=above:{$+2$}] (v2) at (2,2) {};
        \node[vertex_small, label=above:{$\dots$}] (v3) at (3,2) {};
        \node[vertex_small] (v4) at (4,2) {};
        \node[vertex_small, label=above:{$+x_i$}] (v5) at (5,2) {};
        \node[vertex_small] (v6) at (6,2) {};
        \node[vertex_small] (v7) at (7,2) {};
        \node[vertex_small] (v8) at (8,2) {};
        \node[vertex_small] (v9) at (9,2) {};
        
        \draw[-{Stealth[length=5mm,width=5mm]}, line width=1.3mm, lightgray, text=black] (-0.3,3.5) -- (1.5,3.5)
        node[midway, above=5pt] {$n - x_i$ times} ;
        \draw[-{Stealth[length=5mm,width=5mm]}, line width=1.3mm, lightgray, text=black] (10.3,3.5) -- (8.5,3.5)
        node[midway, above=5pt] {$n - (n_i - x_i)$ times} ;

        \draw[thick] (a)--(v1)--(v2)--(v3)--(v4)--(v5) (v6)--(v7)--(v8)--(v9)--(b);
        \draw[thick] (10.3,2.8)--(b)--(10.3,1.2) (-0.3,2.8)--(a)--(-0.3,1.2);
        \draw[thick, dashed] (v5)--(v6);

        \draw[decorate, decoration={brace, mirror, amplitude=5pt}]
        (0.9,1.5) -- (5.1,1.5)
        node[midway, below=5pt] {\Large$x_i$};
        \draw[decorate, decoration={brace, mirror, amplitude=5pt}]
        (5.9,1.5) -- (9.1,1.5)
        node[midway, below=5pt] {\Large$n_i - x_i$};

    \end{tikzpicture}
    \caption{The figure sketches the meaning of the variables and the function~$h$ in the IQP in the proof of \cref{thrm:ml_fpt}.
    (i) It shows a trail~$P_i$ from~$\mathcal{S}$ with endpoints~$a_i$ and~$b_i$.
    The trail has~$n_i = 9$ internal vertices and setting the variable value~$x_i = 5$ corresponds to removing the~$x_i + 1 = 6$th edge counting from~$a_i$.
    In particular, a path of length~$x_i$ is attached to~$a_i$ and a path of length~$n_i - x_i$ is attached to~$b_i$.
    (ii) The function~$h$ computes the distances traveled inside the two obtained paths.
    In particular, for the~$n-x_i$ vertices, which are not in the path attached to~$a_i$, we have to add the distances~$1+2+ \dots +x_i$.
    Similarly for the endpoint~$b_i$.
    Additionally, the function~$h$ also includes the distances inside the paths attached to~$a_i$ and~$b_i$ respectively.}
    \label{fig:ml_iqp}
\end{figure}

First, we consider how many times the distances between the vertices in~$R_\mathcal{S}$ and the endpoints of our trails in~$\mathcal{S}$ have to be added.
Let~$u \in R_\mathcal{S}$ and consider trail~$P_i$.
We have to add~$x_i$ times the distance between~$u$ and~$a_i$ and~$(n_i - x_i)$ times the distance between~$u$ and the endpoint~$b_i$.
This yields:
\begin{align*}
    f(x_1, \dots, x_k) \coloneq \sum_{i=1}^{k} \sum_{u \in R_\mathcal{S}} x_id_{\mathcal{S}}(u,a_i) + (n_i - x_i) d_{\mathcal{S}}(u,b_i)
\end{align*}

Second, we consider how often we need to add the distances between endpoints of trails in~$\mathcal{S}$ to each other.
For each trail~$P_i$, we have to add~$x_i(n_i - x_i)$ times the distances between~$a_i$ and~$b_i$.
For each pair of different trails~$P_i \neq P_j$, we have to add~$x_ix_j$ times the distance between~$a_i$ and~$a_j$, $x_i(n_j - x_j)$ times the distance between~$a_i$ and~$b_j$, $(n_i - x_i)x_j$ times the distance between~$b_i$ and~$a_j$ and $(n_i - x_i)(n_j - x_j)$ times the distance between~$b_i$ and~$b_j$.
\begin{align*}
g(x_1, \dots,&x_k) \coloneq
    \sum_{i=1}^{k} x_i(n_i - x_i)d_{\mathcal{S}}(a_i,b_i)
    +\sum_{j=i+1}^{k} \Bigl( x_ix_jd_{\mathcal{S}}(a_i,a_j) \\
      &+ x_i(n_j - x_j)d_{\mathcal{S}}(a_i,b_j)
      + (n_i - x_i)x_jd_{\mathcal{S}}(b_i,a_j)
      + (n_i - x_i)(n_j - x_j)d_{\mathcal{S}}(b_i,b_j)\Bigr)
\end{align*}

Third, it remains to consider the distance traveled inside the two paths which are obtained by removing an edge from a trail~$P_i$.
Here, it is clear that we have to add~$(n-x_i)$ times the sum of all distances from~$a_i$ to the~$x_i$ vertices from the trail~$P_i$ that are attached to~$a_i$.
As these vertices form a path, the sum of these distances is~$\sum_{j=1}^{x_i} j = \frac{x_i(x_i + 1)}{2}$.
Analogous for the endpoint~$b_i$ with~$(n_i -x_i)$ instead of~$x_i$.
Finally, we have to add the distances inside the paths of length~$x_i$ and~$n_i-x_i$.
Using the fact that a path on~$n$ vertices has Wiener index~$\frac{n(n^2 - 1)}{6}$, we obtain:
\begin{align*}
    h(x_1, \dots, x_k) \coloneq & \sum_{i=1}^{k} \biggl((n-x_i) \frac{x_i(x_i + 1)}{2} + (n - (n_i -x_i))\frac{(n_i - x_i) (n_i - x_i + 1)}{2}\\
    & + \frac{x_i(x_i^2 - 1)}{6} + \frac{(n_i - x_i)\big((n_i - x_i)^2 - 1\big)}{6} \biggr).
\end{align*}
Note that this is indeed a quadratic polynomial since the cubic terms of the products add up to zero.
Adding the constraints that the values of~$x_i$ are bounded by the number of internal vertices in~$P_i$ yields the following IQP formulation:
\begin{align*}
    \min_{(x_1, \dots, x_k) \in \Z^k} &f(x_1, \dots, x_k) + g(x_1, \dots, x_k) + h(x_1,\dots,x_k) \\
    \text{subject to } \ \ & 0 \leq x_i \leq n_i ,  \ \ i \in [k]
\end{align*}

In the final step of the algorithm we apply a result by \citet{Wei2026} which states that \textsc{Integer Quadratic Programming} is solvable in~$f(k) \cdot L^{g(k)}$ time where~$k$ is the number of variables and~$L$ is the total binary length of the input.
As our IQP has~$2k$ constraints and each coefficient is of order~$O(n^2)$, it follows that our IQP instance can be solved in~$O\big(f(k) \cdot (k^2 \log n)^{g(k)}\big)$ time.
It is folklore that~$(\log n)^{g(k)} \in O\big(h(k) \cdot n\big)$ for some computable function~$h$ depending on~$g$.
As~$k \in O(\ml(G)^2)$, we conclude with an overall FPT running time for the parameter~$\ml(G)$.
\end{proof}

To the best of our knowledge, the algorithm from above is the first application of the result by \citet{Wei2026} for parameterized graph algorithms.

\subsection{Above Guarantee Parameterization}
    For a given \MADST-instance~$(G,b)$, we define the parameter~$k \coloneqq b - \W(G)$.
    We next observe that \MADST is in FPT with respect to this ``above guarantee'' parameter.

    \begin{proposition}
		\label{prop:fpt-above-guarantee}
        \MADST is solvable in $(k+2)^kn^{O(1)}$ time.
    \end{proposition}

    \begin{proof}
        Let $(G,b)$ be a \MADST instance and consider the following branching algorithm.
        If~$G$ is a tree, then return ``yes'' if~$b\ge \W(G)$ and
        otherwise return ``no''.
        If~$G$ is not a tree, then find a shortest cycle~$C$ in~$G$.
        If~$C$ has length at least~$k+3$, then return ``no''.
        If~$C$ has length at most~$k+2$, then branch for every edge~$uv$ from~$C$ by calling the algorithm on~$(G -uv,b)$.

        The correctness of the algorithm comes from the fact that a shortest cycle of length at least~$k+3$ implies that~$(G,b)$ is a no-instance since deleting any edge~$uv$ from the cycle increases the distance between~$u$ and~$v$ by at least~$k+1$.

        To see that the algorithm runs in FPT time with respect to~$k$, note that we only branch into at most~$k+2$ cases and that we branch at most~$k$ times since~$\W(G-uv) > \W(G)$ (and hence~$k$ decreases in each branching).
        In particular, the running time is~$(k+2)^kn^{O(1)}$.
    \end{proof}

    If~$b$ equals the actual Wiener index of a \MAD for~$G$, then~$k$, as a graph parameter, is larger than the feedback edge number~$\ell$ (number of edges to remove in order to obtain a tree): 
    Any \MAD has exactly~$\ell$ edges fewer than~$G$ and for any edge~$uv \in E(G) \setminus E(T)$ the distance between~$u$ and~$v$ is at least one larger in~$T$ than in~$G$.
    However, $k$ is still incomparable to other graph parameters such as the vertex cover number or the maximum leaf number.

\section{NP-hardness on Split Graphs}
\label{sec:hardness}
We adapt the construction of \citet{JLK78} to obtain NP-hardness on split graphs, which are graphs that can be partitioned into a clique and an independent set.
Note that split graphs contain no induced~$P_5$.
Thus, the NP-hardness contrasts the polynomial-time algorithm for $P_4$-free graphs (that is, cographs) implied by \Cref{thrm:mw_fpt} (cographs have modular width 2).
\begin{restatable}{theorem}{nphardness}
	\label{thm:np-hardness}
	\MADST is NP-hard on split graphs.
\end{restatable}

\begin{proof}
	We give a polynomial-time reduction from \XtC where each element appears in at most three sets.
	This problem is known to be NP-hard \cite{GJ79}.

	\problemDef{\ExThreeCov (\XtC)}
	{A universe $X=\{x_1, \dots, x_{3q}\}$ and a collection~$\mathcal{C} = \{C_1, \dots, C_s\}$ of size-three subsets of~$X$.}
	{Are there sets~$C_{i_1}, \dots , C_{i_q}$ such that~$\bigcup_{j=1}^q C_{i_j} = X$?}

	Let~$(X = \{x_1, \dots, x_{3q}\}, \mathcal{C} = \{C_1, \dots, C_s\})$ be an \XtC-instance.
	We construct our graph~$G=(V,E)$ as follows (see \Cref{fig:np-hardness}):
		\begin{align*}
		V &\coloneqq X \cup \mathcal{C}\\
		E &\coloneqq \binom{\mathcal{C}}{2} \cup \{x_iC_j \mid i \in [3q], j \in [s], x_i \in C_j\}.
	\end{align*}
	In particular, $G[\mathcal{C}]$ is a clique and~$G[X]$ is an independent set.%

    \begin{figure}[t]
		\centering
		\begin{tikzpicture}
            \node[draw=black, fill=white, inner xsep=1pt, inner ysep=2pt, rounded corners] (s1) at (0,0) {\scriptsize$\{1,2,3\}$};
            \node[draw=black, fill=white, inner xsep=1pt, inner ysep=2pt, rounded corners] (s2) at (1.5,0) {\scriptsize$\{1,3,5\}$};
            \node[draw=black, fill=white, inner xsep=1pt, inner ysep=2pt, rounded corners] (s3) at (3,0) {\scriptsize$\{2,4,6\}$};
            \node[draw=black, fill=white, inner xsep=1pt, inner ysep=2pt, rounded corners] (s4) at (4.5,0) {\scriptsize$\{3,4,6\}$};
            \node[draw=black, fill=white, inner xsep=1pt, inner ysep=2pt, rounded corners] (s5) at (6,0) {\scriptsize$\{4,5,6\}$};

            \node[vertex_small] (1) at (0.5,-2) [label=below:{\scriptsize$1$}] {};
            \node[vertex_small] (2) at (1.5,-2) [label=below:{\scriptsize$2$}] {};
            \node[vertex_small] (3) at (2.5,-2) [label=below:{\scriptsize$3$}] {};
            \node[vertex_small] (4) at (3.5,-2) [label=below:{\scriptsize$4$}] {};
            \node[vertex_small] (5) at (4.5,-2) [label=below:{\scriptsize$5$}] {};
            \node[vertex_small] (6) at (5.5,-2) [label=below:{\scriptsize$6$}] {};

            \draw (s1)--(1) (s1)--(2) (s1)--(3);
            \draw (s2)--(1) (s2)--(3) (s2)--(5);
            \draw (s3)--(2) (s3)--(4) (s3)--(6);
            \draw (s4)--(3) (s4)--(4) (s4)--(6);
            \draw (s5)--(4) (s5)--(5) (s5)--(6);

            \draw (s1)--(s2)--(s3)--(s4)--(s5);
            \draw (s1)to[bend left=30] (s3);
            \draw (s1)to[bend left=30] (s4);
            \draw (s1)to[bend left=30] (s5);
            \draw (s2)to[bend left=30] (s4);
            \draw (s2)to[bend left=30] (s5);
            \draw (s3)to[bend left=30] (s5);

            \draw[very thick] (s2)--(s3) (s1)--(s2) (s2)to[bend left=30](s4) (s2)to[bend left=30](s5) (s2)--(1) (s2)--(3) (s2)--(5) (s3)--(2) (s3)--(4) (s3)--(6);

            \begin{scope}[on background layer]
                \draw[draw=black!30, rounded corners] (-0.7,-0.3) rectangle (6.7,1.2);
                \node at (-1.5,0.5) {Clique};
                \draw[draw=none, fill=black!10, rounded corners] (0,-2.5) rectangle (6,-1.5);
                \node at (-1.5,-2) {Independent Set};
            \end{scope}
        \end{tikzpicture}
		\caption{An illustration of the construction from \Cref{thm:np-hardness}.
		The instance of \XtC has universe $X = \{1, \dots, 6\}$ and the collection of 3-element subsets is $\mathcal{C} = \{\{1,2,3\}, \{1,3,5\}, \{2,4,6\}, \{3,4,6\}, \{4,5,6\}\}$.
		\Cref{lem:Dankelmann} implies that in every \MAD~$T$ there is a set-vertex~$C_i$ such that every path from~$C_i$ in~$T$ is induced in the split graph.
        An example of a MAD tree is highlighted with bold edges in the figure.
        It is associated with the exact cover $\{\{1,3,5\}, \{2,4,6\}\}$ and every path from $\{1,3,5\}$ in the MAD tree is induced in the split graph.}
		\label{fig:np-hardness}
    \end{figure}

	To define the exact value $b$ of the desired Wiener index, we first define the following numbers, where~$D_{AB}$ will later coincide with the sum of distances between the vertex sets~$A$ and~$B$ in a MAD tree.
	\begin{align*}
		D_{\mathcal{C}\mathcal{C}} &\coloneqq (s-1)^2, \\
		D_{\mathcal{C}X} &\coloneqq 3(1 +2(s-1)) + (3q-3)(3 + 3(s-2)), \\
		D_{XX} &\coloneqq 6 + 9(3q-3) + 2(3q-3)(3q-5).
	\end{align*}
	Finally, we set~$b\coloneqq D_{\mathcal{C}\mathcal{C}} + D_{\mathcal{C}X} + D_{XX}$ to obtain the~\MADST-instance~$(G,b)$.
	Next, we show that $(X, \mathcal{C})$ is a yes-instance if and only if there is a spanning tree~$T$ in~$G$ with~$\W(T) \le b$.

	``$\Rightarrow$'': Let~$\mathcal{S} \coloneqq \{C_{i_1}, \dots, C_{i_q}\}$ be an exact cover of~$X$.
	We define the spanning tree~$T$ such that~$T[\mathcal{C}]$ is a star with center~$C_{i_1}$ and such that each $C_{i_j} \in \mathcal{S}$ is adjacent to its three elements.
	Formally,
	\begin{align*}
		E(T) \coloneqq \{C_{i_1}C_j \mid j \in [s], j \neq i_1\} \cup \{x_kC_{i_j} \mid k \in[3q], j \in [q], x_k \in C_{i_j}\}.
	\end{align*}
	Now, it is straightforward to count the distances in the spanning tree~$T$.
	Since~$T[\mathcal{C}]$ is a star, we have
	\begin{align*}
		\dist_T(\mathcal{C},\mathcal{C}) = (s-1)^2 = D_{\mathcal{C}\mathcal{C}}.
	\end{align*}
	A vertex~$x_k \in C_{i_j}$ with~$j \neq 1$ has distance two to the other two elements from~$C_{i_j}$, distance three to the elements from~$C_{i_1}$ and distance four to all others.
	Consequently,
	\begin{align*}
		\dist_T(X,X) &= \dist_T(X\setminus C_{i_1}, X\setminus C_{i_1}) + \dist_T(C_{i_1}, X\setminus C_{i_1}) + \dist_T(C_{i_1}, C_{i_1})\\
		&= \frac{1}{2}(3q-3) (2 \cdot 2 + (3q-6) \cdot 4) + 3 \cdot (3q-3) \cdot 3 + 6 \\
		&= 2(3q-3)(3q-5) + 9(3q-3) + 6  = D_{XX}.
	\end{align*}
	Counting the distances between~$X$ and~$\mathcal{C}$ in the same manner, we obtain
	\begin{align*}
		\dist_T(\mathcal{C}, X) = 3(1 +2(s-1)) + (3q-3)(3 + 3(s-2)) = D_{\mathcal{C}X}
	\end{align*}
	and therefore~$\W(T) = b$.

	``$\Leftarrow$'': Let~$T$ be a \MAD of~$G$ with~$\W(T) \leq b$.
	Note that by construction the closed neighborhood of a vertex~$x_k \in X$ is strictly contained in the closed neighborhood of any vertex~$C_i \in \mathcal{C}$ where~$x_k \in C_i$.
	Hence, by \Cref{lem:Dankelmann}, we can assume that there is a vertex~$C^* \in \mathcal{C}$ for which every~$T$-path starting at~$C^*$ is induced in~$G$.
	Hence, $T[\mathcal{C}]$ is a star with center~$C^*$ and therefore all vertices from~$X$ are leaves in~$T$.
	We denote the three neighbors of~$C^*$ in~$X$ by~$X^*$.
	By simply counting the distances in the star~$T[\mathcal{C}]$, we conclude~$\dist_T(\mathcal{C},\mathcal{C}) = D_{\mathcal{C}\mathcal{C}}$.
	When counting the distances between~$\mathcal{C}$ and~$X$, we distinguish between~$X^*$ and~$X \setminus X^*$.
	This yields
	\begin{align*}
		\dist_T(\mathcal{C},X)
		&= 3(1 +2(s-1)) + (3q-3)(3 + 3(s-2)) = D_{\mathcal{C}X}.
	\end{align*}
	Since~$\W(T) \leq b$, we can also conclude~$\dist_T(X,X) \leq D_{XX}$.
	We now show that in~$T$ each vertex from~$\mathcal{C}$ that has a neighbor in~$X$ has exactly three neighbors in~$X$.
	For this, let $s_\ell$ with~$\ell \in \{1,2,3\}$ be the number of vertices in~$\mathcal{C}\setminus \{C^*\}$ with exactly~$\ell$ neighbors in~$X \setminus X^*$.
	This yields
	\begin{align*}
		\dist_T(X,X) &= \dist_T(X^*,X^*) + \dist_T(X^*,X \setminus X^*) + \dist_T(X \setminus X^*, X \setminus X^*) \\
		&=6 + 3 \cdot 3(3q-3) + \big ( 4\binom{3q-3}{2} -2s_2 -6s_3  \big )\\
		&= 6 + 9(3q-3) + 2(3q-3)(3q-4) - 2s_2 - 6s_3 \\
		&= D_{XX} + 6(q-1) - 2s_2 -6s_3.
	\end{align*}
	Since~$\dist_T(X,X) \leq D_{XX}$, we must have~$6(q-1) - 2s_2 -6s_3 \leq 0$.
	Since every vertex in~$X$ is a leaf in~$T$, we know that~$s_1 + 2s_2 + 3s_3 = 3q-3$ and this together with~$6(q-1) \leq 2s_2  + 6s_3$ implies~$s_1 + 2s_2 + 3s_3 \leq s_2 + 3s_3$.
	Since~$s_1,s_2,s_3 \geq 0$, it follows~$s_1 = s_2 = 0$ and~$s_3 = q-1$.
	This allows us to define an exact cover~$\mathcal{S}$ for the \XtC-instance by setting~$\mathcal{S} \coloneqq \{C_i \in \mathcal{C} \mid |N_T(C_i) \cap X | = 3\}$.
\end{proof}

\section{Discussion and Conclusion}
\label{sec:conclusion}

We initiated a parameterized analysis of \MADSTl, providing several parameterized algorithms.
For future work, we discuss two lines of possible research in more detail.

\subparagraph*{Lower Bounds.}
To the best of our knowledge, there is only one known NP-hardness result for \MADST dating back to 1978~\cite{JLK78}.
We are able to adapt the construction of \citet{JLK78} to obtain NP-hardness on split graphs.
In general, the structure of \MADs makes it difficult not only to design algorithms but also to design reductions.
Notably, our NP-hardness result (as well as the existing one) actually shows that it is NP-hard to find a MAD tree which is a shortest-path tree.
In general, however, a \MAD is not always a shortest-path tree \cite{Da00}. 
Maybe one needs to use such special MAD trees for hardness constructions.
Further lower bounds in terms of hardness results could shed some more light on what makes computing \MADs challenging.

\subparagraph*{Pushing the Border of Tractability.}
Given the long list of algorithmic and structural insights~\cite{DDR04,WLBCRT00,EC85,BE95,DDGS03,JM12,Mondal13,JMP15}, it seems easier to further extend the list of known tractable special cases than providing lower bounds.
From the viewpoint of parameterized algorithmics, the following questions arise:
\begin{compactitem}
    \item Is the problem in FPT for treewidth? %
    As a first step, one might show FPT for the tree-depth. %
    \item Is it in XP for the cliquewidth? (An open question by \citet{DDGS03})
    \item What about other parameterizations such as ``distance to polynomial-time solvability'' (for example, distance to cluster or distance to outerplanar)?
\end{compactitem}

\bibliographystyle{plainnat}
\bibliography{References_arxiv.bib}

@article{Wiener47,
  title={Structural determination of paraffin boiling points.},
  author={Harry Wiener},
  journal={Journal of the American Chemical Society},
  year={1947},
  volume={69},
  number={1}, 
  pages={17--20}
}

@article{Hu74,
  author       = {T. C. Hu},
  title        = {Optimum Communication Spanning Trees},
  journal      = {{SIAM} J. Comput.},
  volume       = {3},
  number       = {3},
  pages        = {188--195},
  year         = {1974},
  doi          = {10.1137/0203015}
}

@article{JLK78,
author = {Johnson, D. S. and Lenstra, J. K. and Kan, A. H. G. Rinnooy},
title = {The complexity of the network design problem},
journal = {Networks},
volume = {8},
number = {4},
pages = {279--285},
url = {https://doi.org/10.1002/net.3230080402},
year = {1978}
}

@article{EC85,
author = {El-Mallah, Ehab S. and Colbourn, Charles J.},
title = {Optimum Communication Spanning Trees in Series-Parallel Networks},
journal = {SIAM Journal on Computing},
volume = {14},
number = {4},
pages = {915--925},
year = {1985},
url = {https://doi.org/10.1137/0214064}
}

@mastersthesis{Schmuck10,
    author      = {Nina Sabine Schmuck},
    title       = {The Wiener index of a graph},
    type        = {Diploma thesis},
    school = {Graz University of Technology},
    year        = {2010}
}

@article{WLBCRT00,
author = {Wu, Bang Ye and Lancia, Giuseppe and Bafna, Vineet and Chao, Kun-Mao and Ravi, R. and Tang, Chuan Yi},
title = {A Polynomial-Time Approximation Scheme for Minimum Routing Cost Spanning Trees},
journal = {{SIAM} Journal on Computing},
volume = {29},
number = {3},
pages = {761--778},
year = {2000},
url = {https://doi.org/10.1137/S009753979732253X}
}

@article{Entringer,
  author={Roger Entringer},
  title={Distance in Graphs: Trees},
  journal={Journal of Combinatorial Mathematics and Combinatorial Computing},
  volume={24},
  pages={65--84},
  year={1997}
}

@article{BES97,
  author       = {Curtis A. Barefoot and
                  Roger C. Entringer and
                  L{\'{a}}szl{\'{o}} A. Sz{\'{e}}kely},
  title        = {Extremal Values for Ratios of Distances in Trees},
  journal      = {Discrete Applied Mathematics},
  volume       = {80},
  number       = {1},
  pages        = {37--56},
  year         = {1997},
  url          = {https://doi.org/10.1016/S0166-218X(97)00068-1}
}

@article{DEG01,
 author={Dobrynin, Andrey A and Entringer, Roger and Gutman, Ivan},
 year={2001},
 title={Wiener Index of Trees: Theory and Applications},
 journal={Acta Applicandae Mathematica},
 pages={211--249},
 volume={66},
 number={3},
 url={https://doi.org/10.1023/A:1010767517079}
}

@article{DDGS03,
  author       = {Elias Dahlhaus and
                  Peter Dankelmann and
                  Wayne Goddard and
                  Henda C. Swart},
  title        = {{MAD} trees and distance-hereditary graphs},
  journal      = {Discrete Applied Mathematics},
  volume       = {131},
  number       = {1},
  pages        = {151--167},
  year         = {2003},
  url          = {https://doi.org/10.1016/S0166-218X(02)00422-5}
}

@article{DDR04,
  author       = {Elias Dahlhaus and
                  Peter Dankelmann and
                  R. Ravi},
  title        = {A linear-time algorithm to compute a {MAD} tree of an interval graph},
  journal      = {Information Processing Letters},
  volume       = {89},
  number       = {5},
  pages        = {255--259},
  year         = {2004},
  url          = {https://doi.org/10.1016/J.IPL.2003.11.009}
}

@article{DH16,
  author       = {P{\aa}l Gr{\o}n{\aa}s Drange and
                  Markus S. Dregi and
                  Pim van 't Hof},
  title        = {On the Computational Complexity of Vertex Integrity and Component
                  Order Connectivity},
  journal      = {Algorithmica},
  volume       = {76},
  number       = {4},
  pages        = {1181--1202},
  year         = {2016},
  url   = {https://doi.org/10.1007/s00453-016-0127-x},
  doi_IGNORE          = {10.1007/S00453-016-0127-X}
}

@article{C89,
  author = {Arthur Cayley},
  title = {A theorem on trees},
  journal = {Quarterly Journal of Pure and Applied Mathematics},
  volume = {XXIII},
  year = {1889},
  pages = {376--378}
}

@book{GJ79,
  author       = {M. R. Garey and
                  David S. Johnson},
  title        = {Computers and Intractability: {A} Guide to the Theory of NP-Completeness},
  publisher    = {W. H. Freeman},
  year         = {1979},
  isbn         = {0-7167-1044-7},
  bibsource    = {dblp computer science bibliography, https://dblp.org}
}

@article{Da00,
  title={A note on MAD spanning trees},
  author={Dankelmann, Peter},
  journal={Journal of Combinatorial Mathematics and Combinatorial Computing},
  volume={32},
  pages={93--96},
  year={2000},
  publisher={CHARLES BABBAGE RESEARCH CENTRE}
}

@article{Abu-AffashCLM24,
  author       = {A. Karim Abu{-}Affash and
                  Paz Carmi and
                  Ori Luwisch and
                  Joseph S. B. Mitchell},
  title        = {Geometric Spanning Trees Minimizing the Wiener Index},
  journal      = {Comput. Geom. Topol.},
  volume       = {3},
  number       = {1},
  pages        = {2:1--2:15},
  year         = {2024},
  url         = {https://doi.org/10.57717/cgt.v3i1.52},
  doi_IGNORE          = {10.57717/CGT.V3I1.52},
  bibsource    = {dblp computer science bibliography, https://dblp.org}
}

@inproceedings{HochuliHW14,
  author       = {Alexandra Hochuli and
                  Stephan Holzer and
                  Roger Wattenhofer},
  editor_IGNORE       = {Magn{\'{u}}s M. Halld{\'{o}}rsson},
  title        = {Distributed Approximation of Minimum Routing Cost Trees},
  booktitle    = {Structural Information and Communication Complexity - 21st International
                  Colloquium ({SIROCCO} 2014)},
  series       = {Lecture Notes in Computer Science},
  volume       = {8576},
  pages        = {121--136},
  publisher    = {Springer},
  year         = {2014},
  url   = {https://doi.org/10.1007/978-3-319-09620-9\_11},
  doi_IGNORE          = {10.1007/978-3-319-09620-9\_11},
  bibsource    = {dblp computer science bibliography, https://dblp.org}
}

@article{FischettiLS02,
  author       = {Matteo Fischetti and
                  Giuseppe Lancia and
                  Paolo Serafini},
  title        = {Exact algorithms for minimum routing cost trees},
  journal      = {Networks},
  volume       = {39},
  number       = {3},
  pages        = {161--173},
  year         = {2002},
  url   = {https://doi.org/10.1002/net.10022},
  doi_IGNORE          = {10.1002/NET.10022},
  bibsource    = {dblp computer science bibliography, https://dblp.org}
}

@article{Ze18,
  author       = {Meirav Zehavi},
  title        = {The k-leaf spanning tree problem admits a klam value of 39},
  journal      = {Eur. J. Comb.},
  volume       = {68},
  pages        = {175--203},
  year         = {2018},
  url   = {https://doi.org/10.1016/j.ejc.2017.07.018},
  doi_IGNORE          = {10.1016/J.EJC.2017.07.018},
  bibsource    = {dblp computer science bibliography, https://dblp.org}
}

@article{HT95,
  author       = {Refael Hassin and
                  Arie Tamir},
  title        = {On the Minimum Diameter Spanning Tree Problem},
  journal      = {Inf. Process. Lett.},
  volume       = {53},
  number       = {2},
  pages        = {109--111},
  year         = {1995},
  url   = {https://doi.org/10.1016/0020-0190(94)00183-Y},
  doi_IGNORE          = {10.1016/0020-0190(94)00183-Y},
  bibsource    = {dblp computer science bibliography, https://dblp.org}
}

@article{LYTN13,
  author       = {Laurent Lyaudet and
                  Paulin Melatagia Yonta and
                  Maurice Tchuent{\'{e}} and
                  Ren{\'{e}} Ndoundam},
  title        = {Distance Preserving Subtrees in Minimum Average Distance Spanning
                  Trees},
  journal      = {Discret. Math. Algorithms Appl.},
  volume       = {5},
  number       = {3},
  year         = {2013},
  url   = {https://doi.org/10.1142/S1793830913500109},
  doi_IGNORE          = {10.1142/S1793830913500109},
  bibsource    = {dblp computer science bibliography, https://dblp.org}
}

@article{ZCFL19,
  author       = {Carlos Armando Zetina and
                  Iv{\'{a}}n A. Contreras and
                  Elena Fern{\'{a}}ndez and
                  Carlos Luna{-}Mota},
  title        = {Solving the optimum communication spanning tree problem},
  journal      = {Eur. J. Oper. Res.},
  volume       = {273},
  number       = {1},
  pages        = {108--117},
  year         = {2019},
  url          = {https://doi.org/10.1016/j.ejor.2018.07.055},
  doi_IGNORE          = {10.1016/J.EJOR.2018.07.055},
  bibsource    = {dblp computer science bibliography, https://dblp.org}
}

@inproceedings{KN18,
  author       = {Stefan Kratsch and
                  Florian Nelles},
  editor_IGNORE= {Yossi Azar and
                  Hannah Bast and
                  Grzegorz Herman},
  title        = {Efficient and Adaptive Parameterized Algorithms on Modular Decompositions},
  booktitle    = {26th Annual European Symposium on Algorithms (ESA 2018)},
  series       = {LIPIcs},
  volume       = {112},
  pages        = {55:1--55:15},
  publisher    = {Schloss Dagstuhl - Leibniz-Zentrum f{\"{u}}r Informatik},
  year         = {2018},
  url   = {https://doi.org/10.4230/LIPIcs.ESA.2018.55},
  doi_IGNORE          = {10.4230/LIPICS.ESA.2018.55},
  bibsource    = {dblp computer science bibliography, https://dblp.org}
}

@article{CDP19,
  author       = {David Coudert and Guillaume Ducoffe and Alexandru Popa},
  title        = {Fully Polynomial {FPT} Algorithms for Some Classes of Bounded Clique-width Graphs},
  journal      = {{ACM} Transactions on Algorithms},
  volume       = {15},
  number       = {3},
  pages        = {33:1--33:57},
  year         = {2019},
  url   = {https://doi.org/10.1145/3310228},
  doi_IGNORE          = {10.1145/3310228},
  bibsource    = {dblp computer science bibliography, https://dblp.org}
}

@inproceedings{GLO13,
  author       = {Jakub Gajarsk{\'{y}} and Michael Lampis and Sebastian Ordyniak},
  title        = {Parameterized Algorithms for Modular-Width},
  booktitle    = {8th International Symposium on Parameterized and Exact Computation ({IPEC} 2013)},
  series       = {Lecture Notes in Computer Science},
  volume       = {8246},
  pages        = {163--176},
  publisher    = {Springer},
  year         = {2013},
  url   = {https://doi.org/10.1007/978-3-319-03898-8\_15},
  doi_IGNORE          = {10.1007/978-3-319-03898-8\_15},
  bibsource    = {dblp computer science bibliography, https://dblp.org}
}

@inproceedings{TCHP08,
  author       = {Marc Tedder and
                  Derek G. Corneil and
                  Michel Habib and
                  Christophe Paul},
  editor_IGNORE= {Luca Aceto and
                  Ivan Damg{\aa}rd and
                  Leslie Ann Goldberg and
                  Magn{\'{u}}s M. Halld{\'{o}}rsson and
                  Anna Ing{\'{o}}lfsd{\'{o}}ttir and
                  Igor Walukiewicz},
  title        = {Simpler Linear-Time Modular Decomposition Via Recursive Factorizing
                  Permutations},
  booktitle    = {Automata, Languages and Programming, 35th International Colloquium ({ICALP} 2008)},
  series       = {Lecture Notes in Computer Science},
  volume       = {5125},
  pages        = {634--645},
  publisher    = {Springer},
  year         = {2008},
  url          = {https://doi.org/10.1007/978-3-540-70575-8\_52},
  doi_IGNORE          = {10.1007/978-3-540-70575-8\_52},
  bibsource    = {dblp computer science bibliography, https://dblp.org}
}

@article{Lokshtanov15,
  author       = {Daniel Lokshtanov},
  title        = {Parameterized Integer Quadratic Programming: Variables and Coefficients},
  year         = {2015},
  journal      = {CoRR},
  url          = {http://arxiv.org/abs/1511.00310}
}

@article{MNSS19,
  author       = {Adriano Masone and
                  Maria Elena Nenni and
                  Antonio Sforza and
                  Claudio Sterle},
  title        = {The Minimum Routing Cost Tree Problem -- State of the art and a core-node
                  based heuristic algorithm},
  journal      = {Soft Computing},
  volume       = {23},
  number       = {9},
  pages        = {2947--2957},
  year         = {2019},
  url   = {https://doi.org/10.1007/s00500-018-3557-3},
  doi_IGNORE          = {10.1007/S00500-018-3557-3}
}

@article{JM12,
    author = {Biswanath Jana and Sukumar Mondal},
    title = {Computation of a Minimum Average Distance Tree on
Permutation Graphs},
    journal = {Annals of Pure and Applied Mathematics},
    volume = {2},
    number = {1},
    pages = {74--85},
    year = {2012}
}

@article{Mondal13,
    title = {An Efficient Algorithm for Computation of a
Minimum Average Distance Tree on Trapezoid
Graphs},
    author = {Sukumar Mondal},
    journal = {Journal of Scientific Research \& Reports},
    volume = {2},
    number  = {2},
    pages = {598--611},
    year = {2013}
}

@article{JMP15,
    title={Computation of a Minimum Average Distance Tree on
Circular-arc Graphs},
    author = {Biswanath Jana and Sukumar Mondal and Madhumangal Pal},
    journal = {International Journal of Electronics Communication and Computer Engineering},
    volume = {6},
    number = {3},
    pages = {384--390},
    year = {2015}
}

@article{MP88,
  title={How to compute the Wiener index of a graph},
  author={Bojan Mohar and Tomaz Pisanski},
  journal={Journal of mathematical chemistry},
  year={1988},
  volume={2},
  issue={3},
  pages={167--277}
}

@inproceedings{BE95,
    author={Burns, K. and Entringer, R. C.},
    title={A graph-theoretic view of the United States postal service},
    editor_IGNORE={Y. Alavi and A. J. Schwenk},
booktitle={Graph Theory, Combinatorics, and Algorithms: Proceedings
of the Seventh International Conference on the Theory and Applications of Graphs},
    publisher={Wiley},
    year = {1995},
    pages={323--334}
}

@book{CFKLMPPS15,
  author    = {Marek Cygan and Fedor V. Fomin and Lukasz Kowalik and Daniel Lokshtanov and D{\'{a}}niel Marx and Marcin Pilipczuk and Michal Pilipczuk and Saket Saurabh},
  title     = {Parameterized Algorithms},
  publisher = {Springer},
  year      = {2015},
  url       = {https://doi.org/10.1007/978-3-319-21275-3},
  bibsource = {dblp computer science bibliography, https://dblp.org}
}

@article{HZ96,
title = {Shortest shortest path trees of a network},
journal = {Discrete Applied Mathematics},
volume = {65},
number = {1},
pages = {275--284},
year = {1996},
url = {https://doi.org/10.1016/0166-218X(95)00038-S},
author = {Pierre Hansen and Maolin Zheng}
}

@book{DF13,
  author       = {Rodney G. Downey and
                  Michael R. Fellows},
  title        = {Fundamentals of Parameterized Complexity},
  publisher    = {Springer},
  year         = {2013},
  url          = {https://doi.org/10.1007/978-1-4471-5559-1}
}

@article{W86,
  author       = {Pawel Winter},
  title        = {An Algorithm for the Enumeration of Spanning Trees},
  journal      = {{BIT}},
  volume       = {26},
  number       = {1},
  pages        = {44--62},
  year         = {1986},
  url          = {https://doi.org/10.1007/BF01939361},
  doi_IGNORE          = {10.1007/BF01939361},
  bibsource    = {dblp computer science bibliography, https://dblp.org}
}

@article{GM78,
author = {Gabow, Harold N. and Myers, Eugene W.},
title = {Finding All Spanning Trees of Directed and Undirected Graphs},
journal = {SIAM Journal on Computing},
volume = {7},
number = {3},
pages = {280-287},
year = {1978},
doi_IGNORE = {10.1137/0207024},
URL = {https://doi.org/10.1137/0207024}
}

@inproceedings{conferenceVersion,
  author       = {Tom{-}Lukas Breitkopf and
                  Vincent Froese and
                  Anton Herrmann and
                  Andr{\'{e}} Nichterlein and
                  Camille Richer},
  editor_IGNORE       = {Florent Foucaud and
                  Aline Parreau},
  title        = {Parameterized Algorithms for Computing {MAD} Trees},
  booktitle    = {Combinatorial Algorithms - 37th International Workshop ({IWOCA} 2026)},
  series       = {Lecture Notes in Computer Science},
  volume       = {16587},
  pages        = {161--174},
  publisher    = {Springer},
  year         = {2026},
  url          = {https://doi.org/10.1007/978-3-032-27732-9\_12},
  doi_IGNORE          = {10.1007/978-3-032-27732-9\_12}
}

@inproceedings{HanakaOOV26,
  author       = {Tesshu Hanaka and
                  Yuto Okada and
                  Yota Otachi and
                  Lena Volk},
  title        = {Finding a Maximum Common (Induced) Subgraph: Structural Parameters
                  Revisited},
  booktitle    = {{WALCOM:} Algorithms and Computation - 20th International Conference
                  and Workshops on Algorithms and Computation ({WALCOM} 2026)},
  series       = {Lecture Notes in Computer Science},
  volume       = {16444},
  pages        = {172--186},
  publisher    = {Springer},
  year         = {2026},
  url          = {https://doi.org/10.1007/978-981-95-7127-7\_12},
  doi_IGNORE          = {10.1007/978-981-95-7127-7\_12},
}

@article{BankevichKarpov2012,
  author  = {Bankevich, A. V. and Karpov, D. V.},
  title   = {Bounds of the Number of Leaves of Spanning Trees},
  journal = {Journal of Mathematical Sciences},
  volume  = {184},
  number  = {5},
  pages   = {564--572},
  year    = {2012},
  doi_IGNORE     = {10.1007/s10958-012-0881-5},
  url     = {https://doi.org/10.1007/s10958-012-0881-5}
}

@misc{Wei2026,
  author={Wei, Linchuan},
  title={Integer quadratic programming in fixed dimension is polynomial-time solvable},
  howpublished={Optimization Online preprint},
  year={2026},
  month={September},
  url={https://optimization-online.org/?p=36744}}

\newpage

\appendix

\end{document}